%% file: main.tex
\documentclass[conference]{IEEEtran}
\IEEEoverridecommandlockouts
\usepackage[usenames,dvipsnames]{xcolor} 
 
\usepackage[dvipsnames]{xcolor}
\usepackage{soul}

\usepackage{algorithm}  
\usepackage{algorithmicx}  
\usepackage[noend]{algpseudocode}
\usepackage{amsmath}

\algnewcommand\And{\textbf{and} }

\usepackage[noadjust]{cite}

\usepackage{amsmath,amssymb,amsfonts}
\usepackage{mathabx}
\usepackage{mathtools}
\usepackage{graphicx}
\usepackage{textcomp}
\usepackage{url}
\usepackage{caption, subcaption}
\usepackage{hyperref}
\usepackage{svg}
\usepackage[normalem]{ulem}
\useunder{\uline}{\ul}{}
\usepackage{float} 
\usepackage{listings}[language=C++]
\usepackage{multirow}
\usepackage{enumitem}

\usepackage{amsthm}
\usepackage[utf8]{inputenc}
\usepackage[english]{babel}
\theoremstyle{plain}
\newtheorem{theorem}{Theorem}

\newtheorem{lemma}[theorem]{Lemma}

\theoremstyle{definition}
\newtheorem{definition}{Definition}

\usepackage{listings}
\usepackage{xcolor}

\definecolor{codegreen}{rgb}{0,0.6,0}
\definecolor{codegray}{rgb}{0.5,0.5,0.5}
\definecolor{codepurple}{rgb}{0.58,0,0.82}
\definecolor{backcolour}{rgb}{0.95,0.95,0.92}

\lstdefinestyle{mystyle}{
    backgroundcolor=\color{backcolour},   
    commentstyle=\color{codegreen},
    keywordstyle=\color{magenta},
    numberstyle=\tiny\color{codegray},
    stringstyle=\color{codepurple},
    basicstyle=\ttfamily\scriptsize,
    breakatwhitespace=false,         
    breaklines=true,                 
    captionpos=b,                    
    keepspaces=true,                 
    numbers=left,                     
    numbersep=5pt,                  
    showspaces=false,                
    showstringspaces=false,
    showtabs=false,                  
    tabsize=2
}

\usepackage{pifont}

\title{Unleashing the Power of Preemptive Priority-based Scheduling for Real-Time GPU Tasks}

\author{
    \IEEEauthorblockN{
        Yidi Wang, 
        Cong Liu,
        Daniel Wong,
        and Hyoseung Kim}
    \IEEEauthorblockA{
        University of California, Riverside\\
        ywang665@ucr.edu, 
        congl@ucr.edu,
        danwong@ucr.edu,
        hyoseung@ucr.edu
    }
}

\begin{document}
\pagestyle{plain}
\pagenumbering{arabic}

\maketitle

\begin{abstract}
Scheduling real-time tasks that utilize GPUs with analyzable guarantees poses a significant challenge due to the intricate interaction between CPU and GPU resources, as well as the complex GPU hardware and software stack. While much research has been conducted in the real-time research community, several limitations persist, including the absence or limited availability of preemption, extended blocking times, and/or the need for extensive modifications to program code. In this paper, we propose two novel techniques, namely the kernel thread and IOCTL-based approaches, to enable preemptive priority-based scheduling for real-time GPU tasks. Our approaches exert control over GPU context scheduling at the device driver level and enable preemptive GPU scheduling based on task priorities. The kernel thread-based approach achieves this without requiring modifications to user-level programs, while the IOCTL-based approach needs only a single macro at the boundaries of GPU access segments. In addition, we provide a comprehensive response time analysis that takes into account overlaps between different task segments, mitigating pessimism in worst-case estimates. Through empirical evaluations and case studies, we demonstrate the effectiveness of the proposed approaches in improving taskset schedulability and timeliness of real-time tasks. The results highlight significant improvements over prior work, with up to 40\% higher schedulability, while also achieving predictable worst-case behavior on Nvidia Jetson embedded platforms.
\end{abstract}

\section{Introduction}

Real-time cyber-physical systems with GPU workloads have become increasingly prevalent in various domains including self-driving cars, autonomous robots, and edge computing nodes. This trend has been accelerated in recent years by the demand for learning-enabled components as most of their implementations heavily rely on the GPU stack. The scheduling problem of GPU-using tasks in these systems is therefore crucial to ensure timely execution and to meet stringent timing requirements. One of the key challenges here is effectively supporting prioritization and preemption, allowing higher-priority tasks to interrupt and temporarily suspend lower-priority GPU tasks whenever needed. This is particularly important in scenarios where critical high-priority tasks with stringent deadlines need to access GPU resources, while low-priority and best-effort tasks can tolerate such preemption to accommodate their execution. 

As of yet, the default scheduling policy of commercial GPU devices provides little control over the prioritization and preemption of GPU tasks, causing unpredictable task response time and instability in real-time systems. 
The real-time research community has recognized this issue since the early era of GPU computing and has proposed several solutions. In particular, the use of real-time synchronization protocols, such as MPCP~\cite{rajkumar1990real,patel2018analytical} and FMLP+~\cite{BB2014-FMLP+}, has been recognized as a promising way to manage GPU tasks in real-time systems with strong analyzable guarantees on the worst-case task response time. However, these approaches can suffer from long blocking time and priority inversion by lower-priority tasks since GPU access segments are handled non-preemptively. There have been attempts to support priority-based GPU scheduling with preemption capabilities~\cite{Kato2011_RGEM, Basaran2012, Zhou2015}, but they require significant modifications to GPU access code, lack analytical support, and more importantly, may not work properly if the system has processes with unmodified GPU code or graphics applications due to the time-shared GPU context switching behavior of the device driver~\cite{capodieci2018deadline,Bakita2023}. 

In this paper, we address the aforementioned challenges and limitations by proposing novel preemptive priority-based GPU scheduling approaches for real-time GPU task execution in multi-core systems with analyzable guarantees. 
Our work focuses on Nvidia GPUs, especially those on Tegra System-on-Chips (SoCs) used in embedded platforms like Jetson Xavier and Orin. 
We propose two distinct approaches: kernel thread and IOCTL-based approaches, each offering unique advantages with different performance implications. 
These approaches work at the device driver level, and unlike existing techniques, they can protect the execution of real-time GPU processes from interference from best-effort non-real-time CUDA processes and graphics processes in the system. Specifically, the kernel-thread approach requires no modifications to user-level GPU code (both host and kernel code) at all, making it amenable to use with any type of workloads. This is particularly appealing to recent machine learning and computer vision applications as they are built on top of massive libraries that involve hundreds of different kernels. The IOCTL-based approach, on the other hand, requires a small modification to GPU access code, i.e., adding just one macro at the boundaries of GPU segments, but provides more fine-grained and efficient control of the GPU. 
Thanks to the strictly preemptive and priority-driven GPU scheduling behavior, both approaches are analyzable and allow us to derive response-time tests for schedulabiltiy analysis.


In summary, the paper makes the following contributions:
\begin{itemize}
    \item We propose the kernel-thread and IOCTL-based approaches, each of which enables preemptive priority-driven GPU scheduling in a multi-core system equipped with an Nvidia GPU. We give details of their implementations and discuss their runtime characteristics. 
    \item While it is important to run GPU segments according to their original task priority (esp. when task priority is assigned based on criticality), we find that assigning different priorities to GPU segments can yield a significant benefit in taskset schedulability. Our work allows this.
    \item We present a comprehensive analysis on the worst-case task response time under our two proposed approaches. In particular, our analysis for the IOCTL-based approach considers self-suspension and busy-waiting modes during GPU kernel execution and reduces pessimism by taking into account the overlaps between different task execution segments.
    \item Our work is implemented on the latest Nvidia Tegra driver and will be open sourced. Experimental results show that our approaches bring substantial benefits in taskset schedulability compared to previous synchronization-based approaches. A case study on Jetson Xavier and Orin platforms demonstrates the effectiveness of our work over the default GPU driver and the applicability for various generations of architectures.
\end{itemize}

\section{Background on Tegra GPU Scheduling}\label{sec:background}

Computational GPU workloads for Nvidia GPUs are often programmed using the CUDA library. These workloads are represented in \textit{kernels} and user-level processes can launch kernels to the GPU at runtime. CUDA provides processes with \textit{streams} to enable concurrent execution of kernels with a limited number of stream priority levels, e.g., only 2 in the Pascal architecture~\cite{xiang2019pipelined}. 
Since streams are bound to a user-level process that created them, the effect of stream scheduling and stream priority assignment is exerted only within each process boundary. 
The CUDA library is not a must for processes to access the GPU hardware. There are other low-level libraries for general-purpose GPU computing and graphics applications such as OpenCL and Vulkan. Programs built using different libraries co-exist in the system and they send GPU commands to the device driver.

At the device driver level, each process is associated with a \textit{GPU context}, which represents a virtual address space and other runtime states on the GPU side. Any process accessing the GPU has a separate GPU context, regardless of whether it uses the CUDA library or not in the user space, and GPU contexts from different processes are time-sliced to share the GPU hardware.
To ensure fairness and prevent resource contention, the Tegra GPU driver uses a scheduling policy that assigns entries in the ``runlist''.\footnote{In fact, there are multiple runlists but we refer to them as singular for simplicity. When scheduling, the driver cycles through each runlist to handle a higher volume of workloads, and this does not affect our proposed design.} The entries of the runlist represent the allocation of time slices to TSGs (Time-Sliced GPUs) that are directly associated with processes. Each TSG has multiple ``channels'', each of which contains a stream of GPU commands received from its process.

Fig.~\ref{fig:runlist_and_tsgs} illustrates the runlist filled with TSG entries. Each TSG data structure maintains state attributes, such as the process ID (pid) that the TSG is associated with, a list of channels, and the time slice. The runlist is scheduled in a round-robin manner. Hence, for each TSG entry, the GPU executes commands of the corresponding GPU context during the time slice. 
The number of entries for a TSG on the runlist is determined by TSG priority. However, as of this writing, there is no interface provided to configure the time slice length and TSG priority settings from the user space. 


The construction of the runlist in the Tegra GPU driver follows multiple steps. First, processes submit their commands to specific channels associated with their TSGs. Once the commands are submitted, the corresponding TSGs are added to the runlist which is protected by a mutex lock. During the construction of the runlist, TSGs with higher priority are granted a larger time slice and more entries on the runlist.
After construction, the runlist is repeatedly scheduled by the GPU hardware in a round-robin manner. Each entry on the runlist runs for its time slice, and once timeout, the TSG of the next entry is executed. This repeats until all the commands of all active TSGs on the runlist are consumed.

In summary, the Tegra GPU driver employs a \textit{time-sliced round-robin} scheduling approach. This approach, however, does not respect the OS-level scheduling priority of processes, which is the main control knob to tune real-time performance in practice. This causes high-priority real-time tasks to experience unpredictable waiting time when the system accepts new best-effort tasks. In addition, it is not easy for the user to observe such driver-level behavior because GPU profiling tools, such as Nvidia Nsight Systems, do not report GPU context switching events and each kernel execution time appears to be inflated with no time slice information. These issues contribute to difficulties in understanding and predicting the runtime behavior of GPU-enabled real-time systems.

\begin{figure}[]
\vspace{-10pt}
    \centering
    \includegraphics[width=0.7\linewidth]{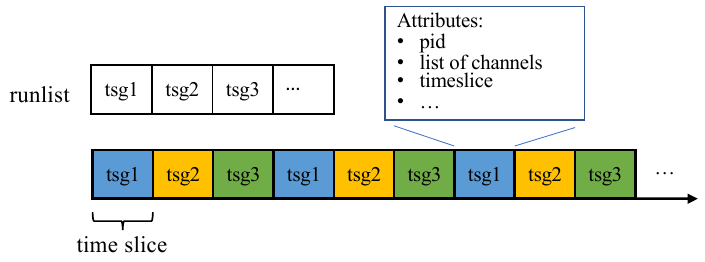}
    \vspace{-0.5\baselineskip}\caption{Runlist and time-sliced GPU scheduling}
    \label{fig:runlist_and_tsgs}
\end{figure}

\begin{table*}[t]
\centering
\begin{tabular}{|ll|c|c|c|c|c|}
\hline
\multicolumn{2}{|l|}{} & No blocking & \begin{tabular}[c]{@{}c@{}}Task priority \\ respected\end{tabular} & \begin{tabular}[c]{@{}c@{}}No source code \\ modification\end{tabular} & \begin{tabular}[c]{@{}c@{}}Analyzable \\ response time\end{tabular} & Inter GPU context \\ \hline
\multicolumn{1}{|l|}{\multirow{2}{*}{\begin{tabular}[c]{@{}l@{}}Prior \\ work\end{tabular}}} & \begin{tabular}[c]{@{}l@{}}Unmanaged GPU  (default driver)\end{tabular}                                                     & \ding{53}   & \ding{53}                                                          & \ding{51}     & \ding{53}  & \ding{51} \\ \cline{2-7} 
\multicolumn{1}{|l|}{}                                                                        & \begin{tabular}[c]{@{}l@{}}Sync.-based  approaches~\cite{Elliott2012,Elliott_RTS13,Elliott_RTSS13,patel2018analytical}\end{tabular} & \ding{53}   & \ding{51}                                                          & \ding{53}     & \ding{51} & \ding{51} \\ \cline{2-7} 
\multicolumn{1}{|l|}{}                                                                        & \begin{tabular}[c]{@{}l@{}}GPU partitioning  approaches~\cite{Bakita2023,Saha2019,wang2021balancing,Jain2019,wang2022towards,zou2023rtgpu,wu2015enabling}\end{tabular} & \ding{51}   & \ding{53}
                 & \ding{53} 
                 & \ding{51}      & \ding{53}                      \\ \cline{2-7}
\multicolumn{1}{|l|}{}                                                                        & \begin{tabular}[c]{@{}l@{}}Preemptive GPU  scheduling~\cite{Kato2011_RGEM,Basaran2012,Zhou2015,capodieci2018deadline,Han2022_reef}\end{tabular} & \ding{51}    & \ding{51} 
                 & \ding{53}    & ?      & \ding{51}                      \\ \hline
\multicolumn{1}{|l|}{\multirow{2}{*}{Ours}}                                                   & \begin{tabular}[c]{@{}l@{}}Kernel thread  approach\end{tabular} & \ding{51}   & \ding{51}                                                          & \ding{51}     & \ding{51}  & \ding{51} \\ \cline{2-7} 
\multicolumn{1}{|l|}{}                                                                        & \begin{tabular}[c]{@{}l@{}}IOCTL-based  approach\end{tabular}   & \ding{51}   & \ding{51}                                                          & \ding{53}     & \ding{51}  & \ding{51} \\ \hline
\end{tabular}
\caption{Comparison of different GPU scheduling approaches
}
\label{tab:comparison_with_prior_work}
\end{table*}

\section{Related Work}\label{sec:related_work}

Table~\ref{tab:comparison_with_prior_work} gives a summary of comparison between representative GPU scheduling approaches. Below we discuss prior work in various categories.

\noindent\textbf{Synchronization-based GPU access control.} Real-time synchronization protocols have played an important role in managing access to GPUs~\cite{Elliott2012,Elliott_RTS13,Elliott_RTSS13,patel2018analytical}. With this approach, GPUs are modeled as mutually-exclusive shared resources and tasks are made to acquire locks to enter code segments accessing the GPUs, i.e., critical sections. MPCP~\cite{rajkumar1990real} and FMLP+~\cite{brandenburg2014fmlp+} are prime examples for multi-core systems with GPUs and the use of such protocols enables analytically provable worst-case task response time bounds.
However, the synchronization-based approach may suffer from blocking time from lower-priority tasks holding a lock and priority inversion caused by the priority boosting mechanism employed in these protocols~\cite{HKim2017}. This becomes particularly problematic when tasks busy-wait on long kernel execution, as discussed in~\cite{patel2018analytical}.

\noindent\textbf{Preemptive GPU scheduling.} Several previous studies~\cite{Kato2011_RGEM, Basaran2012, Zhou2015} have proposed software-based mechanisms to enable preemptive scheduling of real-time GPU tasks. These approaches introduce the concept of decomposing long-running GPU kernels into smaller blocks, allowing preemption to occur at the boundaries of these blocks. By enabling preemptive scheduling, the waiting time of high-priority tasks can be significantly reduced, improving responsiveness and offering a better chance to meet timing requirements. However, the cost of utilizing these mechanisms is not trivial as they necessitate a significant rewriting of user programs~\cite{Basaran2012} or an implementation of a custom CUDA library with device driver modifications~\cite{Basaran2012,Zhou2015}.
Capodieci et al.~\cite{capodieci2018deadline} proposed a hypervisor-based technique to support preemptive Earliest Deadline First (EDF) GPU scheduling of virtual machines (VMs) in a virtualized environment. This approach achieves GPU performance isolation among VMs and shares some similarities with our work, in terms of controlling GPU context switching at the device driver level. However, it lacks consideration of the end-to-end response time of tasks involving CPU and GPU interactions, which is a specific focus of our work.
Recently, Han et al.~\cite{Han2022_reef} proposed REEF, which enables microsecond-scale, reset-based preemption for concurrent DNN inferences on GPUs. This approach proactively kills and restarts best-effort kernels leveraging the idempotent nature of most DNN inference, but it is not applicable to a wide range of applications.

\noindent\textbf{GPU partitioning.} As a GPU is composed of multiple compute units, e.g., Streaming Multiprocessors (SMs) on Nvidia GPUs, there have been attempts to spatially partitioning the GPU and making them accessible by multiple real-time tasks in parallel~\cite{Saha2019,wang2021balancing,Jain2019,wang2022towards,zou2023rtgpu}. They use SM-centric kernel transformation~\cite{wu2015enabling} to run kernels on their designated SMs/partitions. As this involves extensive program modifications and may suffer from misbehaving tasks, Bakita and Anderson~\cite{Bakita2023} recently proposed a user-space library that minimizes program changes and offers much better usability and portability. With GPU partitioning, task performance is greatly affected by partitioning results, e.g., a high-priority task may suffer performance degradation due to the small number of SMs assigned to it or experience blocking if its SMs are shared with other tasks. In addition, 
all these approaches work within a single GPU context, i.e., one process; hence, multiple processes with separate contexts will still time-share the GPU, as discussed in Sec.~\ref{sec:background}. Note that our work does not compete with GPU partitioning techniques. They can be used within each process and our work enables predictable scheduling of GPU processes.



\section{System Model}

We consider a multi-core system with a GPU, which is common in today's embedded hardware platforms like Nvidia Jetson. The CPU has $\omega$ identical cores and the GPU is yet another processing resource used by compute-intensive tasks. The GPU consists of internal resources including Execution Engines (EEs) and Copy Engines (CEs). 
The EE and CE operations of a single process can be done asynchronously at runtime, and during pure GPU execution, the process can either busy-wait or self-suspend on the CPU. However, different processes cannot use the GPU at the same time because of the time-sharing scheduling of GPU contexts at the GPU device driver, as discussed before. 

\smallskip\noindent\textbf{Task Model. }
We consider a taskset $\Gamma$ consisting of $n$ sporadic tasks (processes) with fixed priority and constrained deadlines.\footnote{We assume tasks are processes and use them interchangeably in this paper.} 
Each task is assumed to be preallocated to one CPU core with no runtime migration, i.e., partitioned multiprocessor scheduling.
The execution of a task is an alternating sequence of CPU segments and GPU segments. CPU segments run entirely on the CPU and GPU segments involve GPU operations such as memory copy and kernel execution. 
A task $\tau_i$ can be characterized as follow:
$$    \tau_i := (C_i, G_i, T_i, D_i, \eta^c_i, \eta^g_i) $$

\begin{itemize}
    \item $C_i$: the cumulative sum of the worst-case execution time (WCET) of all CPU segments of task $\tau_i$.
    \item $G_i$: the cumulative WCET of GPU segments (including memory copies and kernels) of $\tau_i$.
    \item $T_i$: the minimum inter-arrival time of each job of $\tau_i$.
    \item $D_i$: the relative deadline of each job of $\tau_i$, and is smaller than or equal to the period, i.e., $D_i \le T_i$.
    \item $\eta^c_i$: the number of CPU segments in each job of task $\tau_i$.
    \item $\eta^g_i$: the number of GPU segments in each job of task $\tau_i$; if $\tau_i$ does not use the GPU, $\eta^g_i=0$.
\end{itemize}
\begin{figure}[h]
\vspace{-10pt}
    \centering
    \includegraphics[width=0.9\linewidth]{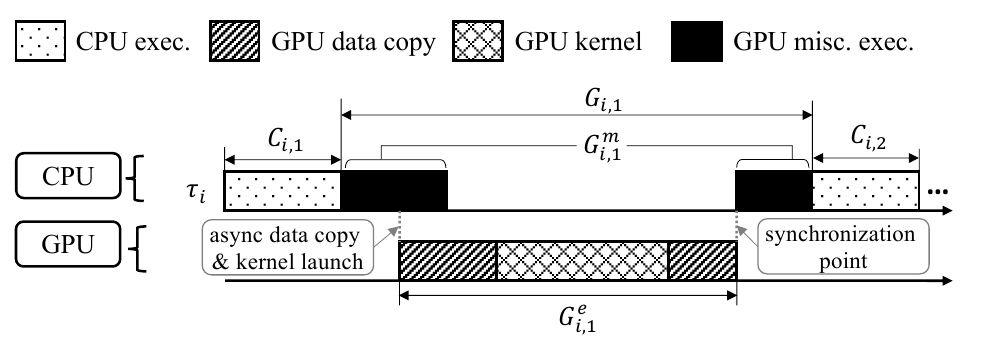}
    \caption{Task model example}
    \label{fig:task_model_orig}
\end{figure}
\vspace{-10pt}
Fig.~\ref{fig:task_model_orig} depicts these parameters, and by default, in each task, all the segments have the same priority. We use $S_{i,j}$ to denote the WCET of the $j$-th segment of type $S$ for task $\tau_i$, where $S$ represents different types of segments, e.g., $G$ for GPU and $C$ for CPU, and $S_i = \sum_{j=1}^{\eta^S_i}S_{i,j}$. Each GPU segment $G_{i,j}$ can be characterized as follow:
$$ G_{i,j} := (G^m_{i,j}, G^e_{i,j}) $$

\begin{itemize}
    \item $G^m_{i,j}$: the cumulative WCET of miscellaneous CPU operations in the $j$-th GPU segment of task $\tau_i$, $G_{i,j}$.
    \item $G^e_{i,j}$: the WCET of GPU workload in $G_{i,j}$ that requires {\em no} CPU intervention; and we call it a {\em pure GPU segment}.
\end{itemize}

$G^m_{i,j}$ is the time for launching a CUDA kernel, overhead for communicating with the GPU driver, and miscellaneous CPU operations for issuing other GPU commands. 
$G^e_{i,j}$ is the time for GPU data copy and kernel execution, during which task $\tau_i$ can either busy-wait or self-suspend on the CPU.
Note that $G_{i,j} \leq G^m_{i,j} + G^e_{i,j}$ because the worst-case of $G^m$ and $G^e$ are not necessarily happening on the same control path and they may execute in parallel in asynchronous mode~\cite{patel2018analytical}.

We also consider best-case execution time, denoted by a check mark (e.g.,  $\widecheck{C_i}$), to improve our analysis in Sec.~\ref{sec:analysis_improved}. For readability, we will explain the parameters that follow this notation where they are used.

\section{Priority-based Preemptive GPU Scheduling}

We present two runtime approaches, kernel thread and IOCTL-based, for preemptive priority-based execution of GPU segments from real-time tasks. The first approach involves a kernel thread that polls for any changes in the status of tasks and updates the runlist accordingly. The second approach involves a set of user-level runtime macros that notify the GPU driver to update the runlist.

The kernel thread approach is easier to use as it does not require any modification to program code, but it updates the runlist only at job execution level, and this may lead to resource underutilization. The IOCTL-based approach provides finer-grained control over GPU segments, but this requires user-level code modification although small. The details of these two approaches are presented in this section.


\begin{algorithm}[t] 
\footnotesize
\begin{algorithmic}[1]  
\Procedure{kernelThreadRunlistUpdate}{}
    \While{true}
        \If{$\tau_i$'s state is changed from the previous cycle} \label{alg1:state_changed_detected}
                \State $\tau_h \gets$ the highest-priority GPU-using ready task \label{alg1:get_tau_h}
                \If{$\tau_h$ exists} \label{alg1:tau_h_exists}
                    \State Add $\tau_h$'s associated TSGs to runlists
                    \State Remove other TSGs from runlists
                \Else \Comment{no \texttt{RUNNING} real-time tasks}
                    \State Add all active TSGs to runlists \label{alg1:done_updating_runlist}
                \EndIf
        \EndIf
        \State Wait for the next polling cycle
    \EndWhile
\EndProcedure
\end{algorithmic}  
\caption{Kernel Thread Approach}
\label{alg:kernel_thread_approach}  
\end{algorithm}

\subsection{Kernel Thread Approach}
\label{sec:kernel_thread_approach}
The kernel thread approach creates a kernel thread that is initiated along with the driver software. It continuously polls for changes in task status (\texttt{task\_struct::state}), with an interval of a sub-scheduling time quantum\footnote{We used 1 ms in our implementation, but the cost of checking task status itself without updating the runlist is negligibly small.} on a designated CPU core, to 
minimize interference to other tasks. When a task status change is detected, e.g., from \texttt{TASK\_RUNNING} to \texttt{TASK\_STOPPED}, it updates the runlist.

The procedure is shown in Alg.~\ref{alg:kernel_thread_approach}. 
At every polling cycle, the kernel thread checks if any task $\tau_i$ with an active TSG has changed its state from the previous cycle (line~\ref{alg1:state_changed_detected}). 
If yes, 
it obtains a ready GPU task (\texttt{TASK\_RUNNING}) with the highest real-time priority (\texttt{task\_struct::rt\_priority}) as $\tau_h$ (line~\ref{alg1:get_tau_h}).\footnote{In this work, we assume that each task has a unique real-time priority. This assumption can be easily achieved by an arbitrary tie-breaking rule in practice, as found in the literature.} Next, the scheduler decides whether the runlist should be updated. If $\tau_h$ exists, the scheduler removes all other TSGs from the runlist but only keeps $\tau_h$'s associated TSGs in it. Otherwise, it means that no GPU-using real-time task is ready to run, and in this case, the scheduler puts all other active TSGs back into the runlist, allowing non-real-time best-effort tasks to make their progress (lines~\ref{alg1:tau_h_exists} to~\ref{alg1:done_updating_runlist}).

However, since the scheduling decision is made only when a task state changes, this approach may underutilize GPU resources. For instance, when a high-priority task $\tau_h$ starts to run, the currently-running task $\tau_i$'s TSGs are removed from the runlist to reserve the GPU for $\tau_h$. However, while $\tau_h$ is executing its CPU segments, the GPU may remain idle, leading to underutilization of GPU resources.
The kernel thread approach has another limitation. Self-suspension during GPU execution is not allowed and a task must spin on the CPU side to maintain its task state. This is because changes in task state due to self-suspension can be misinterpreted by the kernel thread. This may lead to incorrect scheduling decisions and cause unnecessary runlist updates.

\begin{figure}[t]
\centering
    \begin{subfigure}[b]{\linewidth}
    \centering
        \includegraphics[width=0.9\linewidth]{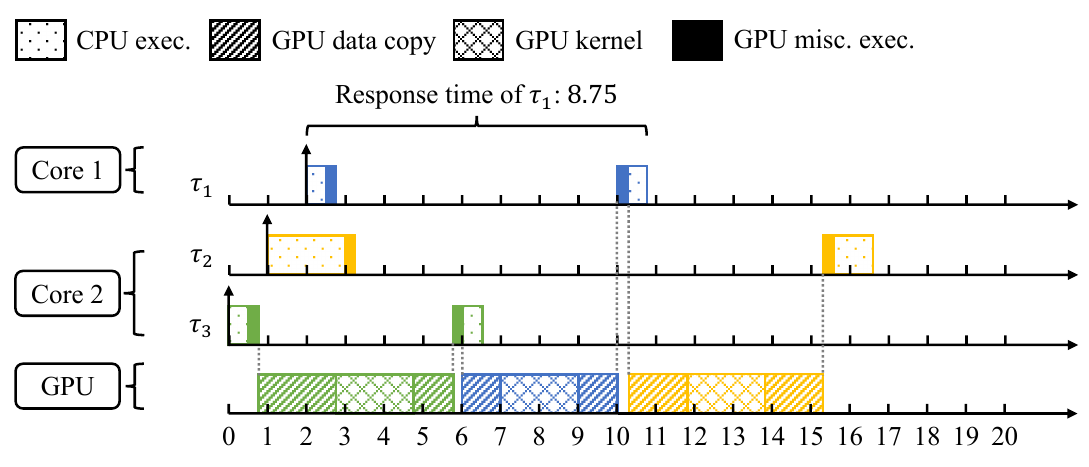}
        \vspace{-0.5\baselineskip}\caption{Schedule under synchronization-based approach}
        \label{fig:sync_based_approach}
    \end{subfigure}
    \begin{subfigure}[b]{\linewidth}
    \centering
        \includegraphics[width=0.9\linewidth]{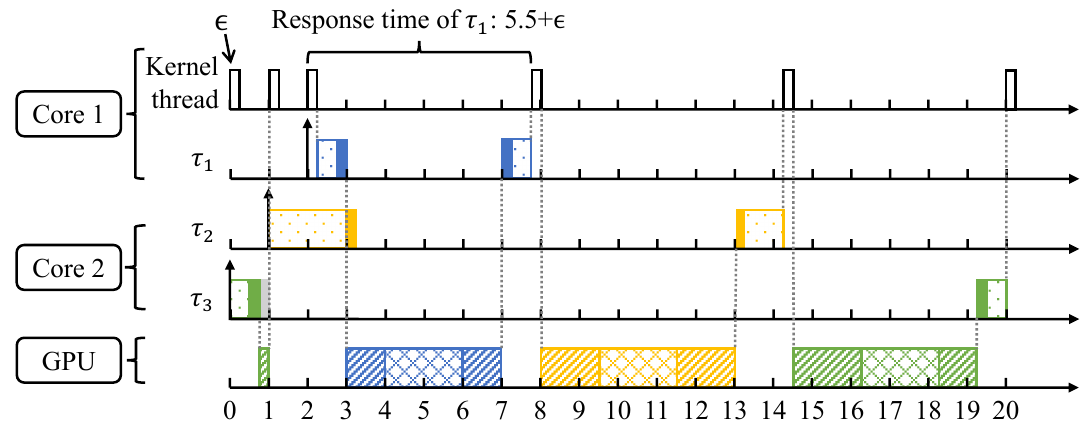}
        \vspace{-0.5\baselineskip}\caption{Schedule under kernel thread approach}
        \label{fig:kthread_approach}
    \end{subfigure}
    \begin{subfigure}[b]{\linewidth}
    \centering
        \includegraphics[width=0.9\linewidth]{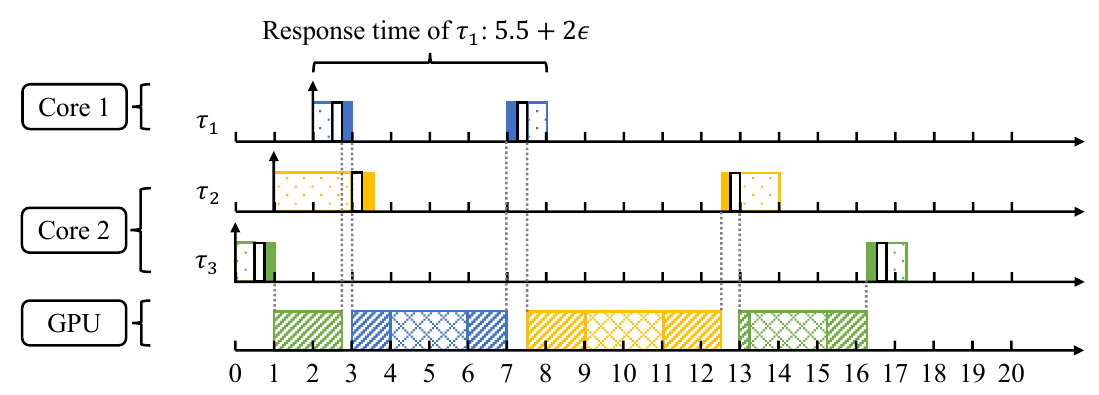}
        \vspace{-0.5\baselineskip}\caption{Schedule under IOCTL-based approach}
        \label{fig:ioctl_approach}
    \end{subfigure}
\caption{Example schedule of three tasks under different approaches (priority $\tau_1 > \tau_2 > \tau_3$)}
\end{figure}

Figs.~\ref{fig:sync_based_approach} and \ref{fig:kthread_approach} compare task schedules under the conventional synchronization-based approach and our kernel thread approach. $\tau_1$ is running on Core 1, while $\tau_2$ and $\tau_3$ are running on Core 2. 
The synchronization-based approach shown in Fig.~\ref{fig:sync_based_approach} treats the entire execution of a GPU segment as a critical section. Tasks are serviced in order based on their task priorities. This approach ensures that each task completes its GPU segments in a deterministic and predictable manner.
However, as can be seen in the figure, $\tau_1$ is delayed by the GPU segments of all of its lower-priority tasks and gets a response time of 8.75. 
On the other hand, our kernel thread approach avoids this delay by allowing preemption during GPU segment execution. In Fig.~\ref{fig:kthread_approach}, the kernel thread is on Core 1 along with $\tau_1$ and the runlist update time is denoted as $\epsilon$.\footnote{Prior work~\cite{capodieci2018deadline} reports that the runlist update overhead including GPU context switching can take from 50 to 750~$\mu$s. Our analysis in Sec.~\ref{sec:analysis} takes into account $\epsilon$ and our measurements in Sec.~\ref{sec:system_eval} show similar results.} At $t=1$, the kernel thread updates the runlist and causes preemption of $\tau_3$'s GPU segment by removing its associated TSGs from the runlist. Task $\tau_1$ is delayed by $\epsilon$ due to running on the same core as the kernel thread. The GPU is then allocated to the highest-priority task, $\tau_1$, until it completes. The response time of $\tau_1$ is 5.5+$\epsilon$, much smaller than that of the synchronization-based approach.

Although it is not depicted in the above figure, there could be a delay for GPU preemption to take effect because Nvidia GPUs support preemption at the pixel level for graphics tasks and the thread-block level for compute tasks~\cite{nvidia_preemption}. Such delay is thus very small compared to the length of GPU kernels, and for compute tasks, it can be separately measured or estimated by the maximum length of a single thread block among all kernels. We assume that $\epsilon$ includes this delay in it.

\subsection{IOCTL-based Approach}
The IOCTL-based approach is a user-level runtime method for efficient control of GPU segments in the runlist. 
To implement this method, we add two macros that allow user programs to indicate the beginning and completion of a GPU segment. When the macro is called, it generates an IOCTL command and sends it to the GPU driver through a file descriptor, and requests the driver to update the runlist accordingly.

\begin{lstlisting}[caption=Example Usage of IOCTL-based approach, label={lst:ioctl_sample}, language=C++]
int task_function() {
    ...
    cudaStreamBegin(fd, getpid());
    cudaMemcpyAsync(d_in, h_in, mem_size_in, cudaMemcpyHostToDevice, stream));
    MyKernel<<<grid, threads, 0, stream>>>(d_in, d_out, dims_in, dims_out);
    cudaStreamEnd(fd, getpid());
    ...
}
\end{lstlisting}

The macros introduced are \texttt{cudaStreamBegin()} and \texttt{cudaStreamEnd()}, which are wrappers to our IOCTL syscalls. A sample user program is listed in Listing~\ref{lst:ioctl_sample}. The code between them is a GPU segment. 
Unlike the kernel thread approach in which the runlist update is triggered at the boundaries of each task's job execution, with the help of these two macros, we can define the boundaries of GPU segments and allows GPU segments and CPU segments from different tasks to be co-scheduled.
In the Tegra driver, the default runlist update is protected by a mutex lock. As the IOCTL-based approach allows multiple tasks to make these calls concurrently, we replace the default mutex with a real-time mutex, i.e., \texttt{rt\_mutex}, to reduce the blocking time as well as prevent priority inversions.

\begin{algorithm}[t] 
\footnotesize
\begin{algorithmic}[1]  
\State $task\_pending = \emptyset$
\State $task\_running = \emptyset$
\State \Comment{Note that a task exclusively exists in one of these two lists}
\Procedure{IOCTLRunlistUpdate}{$\tau_i, add$}        
    \If{$add$} \Comment{$\tau_i$ requests to be added}
        \If{$\tau_i$ is not a real-time task} \label{alg2:tau_i_not_rt}
            \If{no real-time task is in $task\_running$}
                \State Add $\tau_i$ to $task\_running$
            \Else
                \State Add $\tau_i$ to $task\_pending$ \label{alg2:tau_i_not_rt_done}
            \EndIf
        \Else \Comment{$\tau_i$ is a real-time task} \label{alg2:tau_i_rt}
            \State $\tau_h \leftarrow$ the highest-priority task in $task\_running$
            \If{$\tau_i.rt\_priority > \tau_h.rt\_priority$}
                \State Add $\tau_i$ to $task\_running$
                \State Move $\tau_h$ to $task\_pending$
            \Else
                \State Add $\tau_i$ to $task\_pending$ \label{alg2:tau_i_rt_done}
            \EndIf
        \EndIf
    \Else \Comment{$\tau_i$ requests to be removed} \label{alg2:tau_i_to_be_removed}
        \State $\tau_k \leftarrow$ the highest-priority RT task in $task\_pending$
        \If{$\tau_k$ exists}
            \State Move $\tau_k$ to $task\_running$
            \State Remove $\tau_i$ from $task\_running$
        \Else \Comment{no pending real-time task}
            \State $task\_running \gets task\_pending$
            \State $task\_pending \gets \emptyset$ \label{alg2:tau_i_to_be_removed_done}
        \EndIf
    \EndIf
    \State Add all TSGs of tasks in $task\_running$ to the runlist
\EndProcedure
\end{algorithmic}  
\caption{IOCTL-based Approach}
\label{alg:ioctl_approach}  
\end{algorithm}

The procedure to update the runlist under this approach is shown in Alg.~\ref{alg:ioctl_approach}. To track which tasks are in the runlist and which tasks are pending, two bitmaps are maintained in the GPU driver. The procedure is started by an IOCTL command which is wrapped by our macro. 
When a caller task $\tau_i$ requests to be added to the runlist (through \texttt{cudaStreamBegin()}), the scheduler implemented inside the IOCTL function first checks whether $\tau_i$ is a real-time task. If it is not, the scheduler checks whether there is any real-time task that is currently running and decides whether to add $\tau_i$ to the runlist or add it to the pending list (line~\ref{alg2:tau_i_not_rt} to~\ref{alg2:tau_i_not_rt_done}). If $\tau_i$ is a real-time task, the scheduler checks the priority of $\tau_i$ relative to the currently-running task $\tau_h$. If the priority of $\tau_i$ is higher than $\tau_h$, the scheduler preempts the GPU execution of $\tau_h$ and moves it to the pending list, and $\tau_i$ is added to the runlist. Otherwise, $\tau_i$ is added to the pending list (line~\ref{alg2:tau_i_rt} to~\ref{alg2:tau_i_rt_done}).
If $\tau_i$ notifies the driver about the completion through \texttt{cudaStreamEnd()}, the scheduler first finds the highest-priority task $\tau_k$ in the pending list. If $\tau_k$ exists, it is added to the runlists. Otherwise, if there are only best-effort tasks, the scheduler add them to the runlist to resume their progress (line~\ref{alg2:tau_i_to_be_removed} to~\ref{alg2:tau_i_to_be_removed_done}).

Fig.~\ref{fig:ioctl_approach} shows an example schedule under the IOCTL-based approach using the same taskset as in the previous section. Unlike the kernel thread approach, $\tau_3$'s GPU segments are not preempted until $\tau_1$ starts its GPU kernel execution. This strategy is followed in the remaining schedule. 

The IOCTL-based approach allows fine-grained GPU resource control and ensures prompt execution of high-priority tasks. However, user-level code modification is required, which can make the kernel thread approach more appealing. 

\subsection{GPU Segment Priority Assignment}\label{sec:priority_assignment}

In both kernel thread and IOCTL-based approaches, GPU segments are executed following their OS-level task priorities. 
In the kernel thread approach, the kernel thread has the highest priority, and the preemption occurs at job execution boundaries. 
In the IOCTL-based approach, preemption can occur at segment boundaries. 

To improve taskset schedulability, we can assign separate priority to the GPU segments of a task, different from its CPU priority. We adopt Audsley's approach for this purpose~\cite{Audsley2007OPTIMALPA}. Hence, if the schedulability test given in the next section determines a taskset is unschedulable, we iterate through all tasks from the lowest to the highest CPU priority and check whether each priority level can be assigned to the GPU segments of a task without causing the taskset to fail the schedulability test. To prevent deadlocks, we maintain the relative priority order of GPU segments identical to their corresponding CPU segments (i.e., OS-level priorities) for tasks executing on the same CPU core. For instance, consider two tasks $\tau_1$ and $\tau_2$ assigned to the same CPU, with CPU priority $\pi^c_1 > \pi^c_2$. If our algorithm suggests a GPU priority order where $\pi^g_1 < \pi^g_2$, we treat this assignment as infeasible.

\section{End-to-End Response Time Analysis}
\label{sec:analysis}
This section presents a comprehensive analysis on the end-to-end response time of tasks involving CPU and GPU interactions. We first give analysis under two approaches proposed in the previous section. We then introduce a technique to reduce the pessimism of our analysis and provide a tighter bound. 

\subsection{Baseline Analysis}
\label{section:sched_analysis:baseline_analysis}
Our baseline analysis provides a conservative upper bound on the worst-case task response time under the two proposed approaches. In our model, preemption can be divided into two types: (i) {\em CPU preemption}: a CPU segment of a task $\tau_i$ is preempted by a CPU segment of a higher-priority task $\tau_h$ running on the same CPU core, and (ii) {\em GPU preemption}: a GPU segment of a task $\tau_i$ is preempted by a GPU segment of a higher-priority task $\tau_h$, regardless of which CPU core is assigned to $\tau_h$. Based on these, we develop a response-time analysis for the two proposed approaches in the following.

\begin{figure}[t]
\centering
    \begin{subfigure}[b]{\linewidth}
    \centering
        \includegraphics[width=0.9\linewidth]{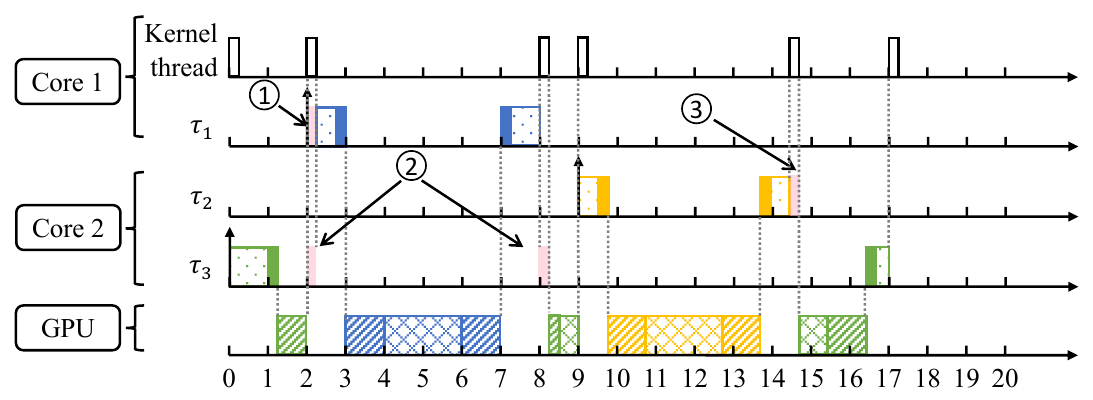}
        \vspace{-0.5\baselineskip}\caption{Kernel thread approach}
        \label{fig:kthread_analysis}
    \end{subfigure}
    \begin{subfigure}[b]{\linewidth}
    \centering
        \includegraphics[width=0.9\linewidth]{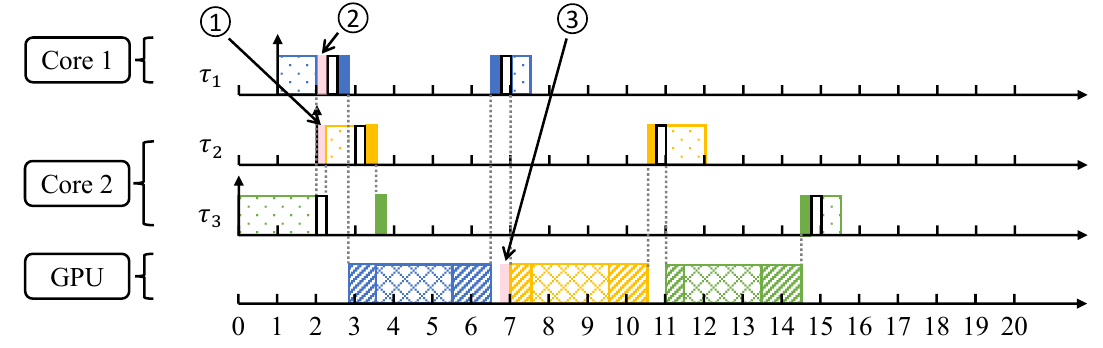}
        \vspace{-0.5\baselineskip}\caption{IOCTL-based approach}
        \label{fig:ioctl_delay}
    \end{subfigure}
\caption{Example schedule of three tasks with runlist update delay (task priority: $\tau_1 > \tau_2 > \tau_3$)}
\end{figure}

\subsubsection{Preemptive GPU under Kernel Thread Approach}
The kernel thread runs on one designated CPU core and updates the runlist on behalf of other GPU-using tasks. 
Task self-suspension is not applicable here, and only the busy-waiting mode is available to monitor task states and prevent incorrect runlist updates in the middle of job execution.

We first identify the delay for a task $\tau_i$ caused by the runlist updates of the kernel thread. 
\begin{lemma} \label{lm:kernel_thread_delay}
The runlist update delay from the kernel thread for a job of task $\tau_i$ is upper-bounded by:
\begin{equation}
\small
\begin{aligned}
    K_i &= x_i \cdot (2\epsilon + \sum_{\tau_h \in hp(\tau_i)\land\eta_h^g>0} \lceil \frac{R_i+J_h}{T_h} \rceil \cdot 2\epsilon)
\end{aligned}
\label{eq:kernel_thread_delay}
\end{equation}
where 
\begin{equation*}
\small
\begin{aligned}
    x_i = 
    \begin{cases}
        1 & \text{, $\tau_i$ is a GPU-using task ($\eta_i^g>0$) or runs on the} \\
        & \text{   same core as the kernel thread}\\
        0 & \text{, otherwise}
    \end{cases}
\end{aligned}
\end{equation*}
, $\epsilon$ is the runlist update time (Sec.~\ref{sec:kernel_thread_approach}), $R_i$ is the worst-case response time of $\tau_i$, $hp(\tau_i)$ is a set of all the higher-priority tasks than $\tau_i$ in the system, and $J_h=R_h-(C_h+G_h)$ is the release jitter to capture the carry-in effect. 
\end{lemma}

\begin{proof}
Whenever the kernel thread updates the runlist, it delays the CPU execution of other tasks on the same core due to its highest priority and also the GPU execution due to TSG evictions and GPU context switching~\cite{capodieci2018deadline}. Hence, GPU-using tasks on any CPU core and CPU-only tasks running on the same core as the kernel thread are subject to this delay. The only exception is CPU-only tasks running on a different core as they are neither delayed by the CPU and GPU operations of the kernel thread (the $x_i$ term). 

Once the job of $\tau_i$ starts execution, its status change triggers the kernel thread once to update the runlist. The kernel thread might be already updating the runlist for a lower-priority task, so one additional $\epsilon$ needs to be considered  (the first term in the parenthesis).
During $\tau_i$'s job execution ($R_i$), additional invocation of the kernel thread is determined by only higher-priority jobs since those with lower priority than $\tau_i$ cannot trigger the runlist update until the completion of the $\tau_i$'s job. The number of arrivals of high-priority tasks during $R_i$ is upper-bounded by $\lceil \frac{R_i+J_h}{T_h} \rceil$, where adding $J_h$ is a known method to capture a carry-in job in an arbitrary interval~\cite{bertogna2008schedulability}. Each high-priority job involves two times of runlist updates ($2\epsilon$), one at the beginning of the high-priority job and another at its completion to resume $\tau_i$'s job.
\end{proof}

Fig.~\ref{fig:kthread_analysis} gives an example of all types of delay caused by the kernel thread. 
\ding{192} is the delay that occurs when any task running on the same core as the kernel thread has a state change, i.e., starting job execution. During the remote preemption of $\tau_1$ on $\tau_3$, \ding{193} illustrates the delay caused by the runlist updates before and after $\tau_1$'s completion. At last, the delay caused by a local preemption from $\tau_2$ on $\tau_3$ is depicted by \ding{194}.

\begin{lemma}\label{lm:baseline_kthread_rt_busy}
Under the kernel thread approach, the worst-case response time of a task $\tau_i$ is upper-bounded by:
\begin{equation}
\small
\begin{aligned}
    R_i &= C_i + G_i + K_i 
        + \sum_{\mathclap{\tau_h \in hpp(\tau_i)}} \quad \lceil \frac{R_i}{T_h} \rceil (C_h + G_h) \\[-2pt]
        &+ \sum_{\mathclap{\substack{\tau_h \in hp(\tau_i)\land \\  \tau_h \notin hpp(\tau_i)\land \eta_h^g>0 }}} \; \lceil \frac{R_i+J_h}{T_h} \rceil (C_h + G_h)
\end{aligned}
\label{eq:baseline_kthread_rt_busy}
\end{equation}
where $hpp(\tau_i)$ is the set of higher-priority tasks running on the same CPU core as $\tau_i$.
\end{lemma}

\begin{proof}
This is an extension of the conventional iterative response-time test. 
$C_i$ and $G_i$ denote the execution time of the CPU and GPU segments of $\tau_i$. 
$K_i$ bounds the delay from the kernel thread. 
By the design of the kernel thread approach, higher-priority tasks $\tau_h$ on the same CPU core as $\tau_h$ can preempt $\tau_i$ for their entire job execution ($C_h + G_h$) with no jitter effect because $\tau_h$ busy-waits on the CPU (the second line of Eq.~\eqref{eq:baseline_kthread_rt_busy}).
Higher-priority GPU-using tasks ($\eta_h^g>0$) on different cores can also effectively preempt $\tau_i$ irrespective of whether $\tau_i$ is a CPU-only task or not due to busy-waiting (the third line). Fig.~\ref{fig:preemption_cpu_task_dft_prio} illustrates an example where $\tau_1$ preempts $\tau_2$'s GPU execution ($G^m$ are omitted for simplicity). As a result, $\tau_2$ remains busy-waiting on the CPU, leading $\tau_3$ to get an additional delay equivalent to the GPU segments of $\tau_1$. Since the release of higher-priority tasks on different cores is not synchronized with $\tau_i$, they can introduce carry-in jobs, captured by $J_h$ in the last term, as in Lemma~\ref{lm:kernel_thread_delay}. 
\end{proof}

\begin{figure}[t]
\centering
    \begin{subfigure}[b]{0.48\linewidth}
        \includegraphics[width=\linewidth]{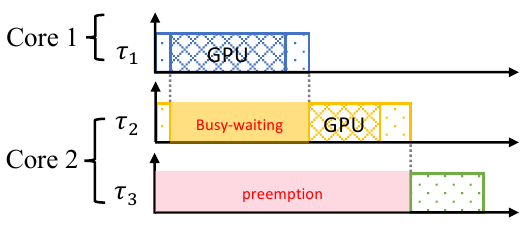}
        \vspace{-1\baselineskip}\caption{CPU priority = GPU priority: $\tau_1 > \tau_2 > \tau_3$}
        \label{fig:preemption_cpu_task_dft_prio}
    \end{subfigure}
    \begin{subfigure}[b]{0.48\linewidth}
        \includegraphics[width=\linewidth]{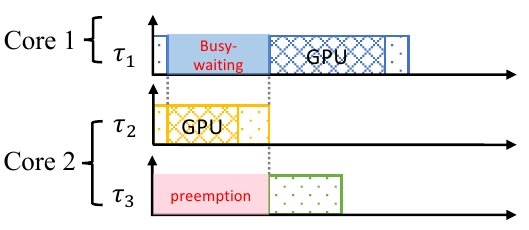}
        \vspace{-1\baselineskip}\caption{CPU priority: $\tau_1 > \tau_2 > \tau_3$; GPU priority: $\tau_2 > \tau_1 > \tau_3$}
        \label{fig:preemption_cpu_task_gpu_prio}
    \end{subfigure}
\caption{Preemption by GPU segments on CPU tasks under busy-waiting mode, kernel thread approach as an example}
\end{figure}

\subsubsection{Preemptive GPU under IOCTL-based Approach}

The IOCTL-based approach provides tasks with both busy-waiting and self-suspension options during pure GPU execution. Unlike the kernel thread approach, each individual task executes the runlist update while holding a lock in the GPU driver for the duration of $\epsilon$ time units. 
For the analysis, the primary differences compared to the kernel thread approach are:
\begin{itemize}
    \item \textit{Runlist update execution}: 
    In the worst case, runlist updates are required both before and after each GPU segment of $\tau_i$, since the associated TSGs need to be added and removed by IOCTL calls as shown in Listing~\ref{lst:ioctl_sample}. This leads to a cumulative cost of $2\epsilon \cdot \eta^g_i$ for the entire job of $\tau_i$.
    \item \textit{Blocking time}: A task $\tau_i$ can experience blocking of $\epsilon$ due to an ongoing runlist update initiated by a lower-priority task.
\end{itemize}
Other factors contributing to the response time under the IOCTL-based approach, such as the execution requirements, i.e., $C_i$ and $G_i$, and possible preemption scenarios, remain consistent with the kernel thread approach. 
For ease of presentation, we define $G_i^*$, $G_i^{e*}$, and $G_i^{m*}$ to incorporate the two times of runlist updates into the execution requirements.
\begin{equation*}
\small
\begin{aligned}
    G^*_i=G_i + 2\epsilon \cdot \eta^g_i\; \text{,}\; G_i^{e*} = G^e_i+2\epsilon\cdot\eta^g_i     
\;\text{and}\; G_i^{m*} = G^m_i+2\epsilon\cdot\eta^g_i
\end{aligned}
\end{equation*}

Before we proceed to analyze the response time, Fig.~\ref{fig:ioctl_delay} can help better understand all types of runlist update delay under the IOCTL-based approach. The task of interest here is $\tau_2$ which runs with medium priority.
The runlist update by $\tau_3$ slightly before $\tau_2$'s release as well as $\tau_1$'s GPU segment causes blocking to $\tau_2$ at \ding{192} and $\tau_1$ at \ding{193}.  
Until $\tau_1$ finishes GPU execution,
$\tau_2$ cannot start its GPU segment as $\tau_1$ is actively using the GPU with higher priority than $\tau_2$. Then, the start time of $\tau_2$'s GPU segment is further delayed by the runlist update at \ding{194}, right after $\tau_1$'s GPU segment, which is to remove $\tau_1$'s TSG from the runlist. 
Based on these observations, we derive the response time analysis of busy-waiting and self-suspending GPU tasks as follows.

\begin{lemma}\label{lm:baseline_ioctl_rt_busy}
Under the IOCTL-based approach with busy-waiting, the worst-case response time of $\tau_i$ is bounded by:
\begin{equation}
\small
\begin{aligned}
    R_i &= C_i + G^*_i + (\eta^g_i+1)\cdot\epsilon \\[-2pt]
        &+ \sum_{\substack{\tau_h \in hpp(\tau_i) \\\land \eta^g_h=0}} \lceil \frac{R_i}{T_h} \rceil \cdot C_h + \sum_{\substack{\tau_h \in hpp(\tau_i) \\\land \eta^g_h>0}} \lceil \frac{R_i}{T_h} \rceil \cdot (C_h + G^*_h) \\[-3pt]
        &+ \sum_{\mathclap{\substack{\tau_h\in hp(\tau_i) \land\\ \tau_h\notin hpp(\tau_i)\land\eta^g_h>0 }}} \quad \lceil \frac{R_i+J^g_h}{T_h} \rceil \cdot G_h^{e*}
\end{aligned}
\label{eq:baseline_ioctl_rt_busy}
\end{equation}
where 
$J^g_h=R_h-G^e_h$.
\end{lemma}
\begin{proof}
The first two terms are self-evident, representing the execution requirement and the runlist update triggered by $\tau_i$ itself, respectively.

The total amount of blocking (third term) imposed by lower-priority tasks is bounded by $(\eta^g_i+1)\cdot\epsilon$. First, at least one time of blocking of $\epsilon$ applies to every $\tau_i$, regardless of whether it is a GPU-using or CPU-only task. This is because it can experience blocking from GPU-using lower-priority tasks at the very beginning of its job instance, as illustrated by \ding{192} in Fig.~\ref{fig:ioctl_delay}. If $\tau_i$ is a GPU-using task, each GPU segment requires up to $\epsilon$ for potential blocking from lower-priority tasks, which results in $\eta^g_h\cdot\epsilon$.

The next factor to consider is preemption by higher-priority tasks $\tau_h$. We divide this into (i) those on the same core as $\tau_i$ ($\tau_h\in hpp(\tau_i)$ in the fourth and fifth terms) and (ii) different cores ($\tau_h\in hp(\tau_i)\land \tau_h\notin hpp(\tau_i)$ in the last term). Let us consider case (i) first. If $\tau_h$ is a CPU-only task ($\eta^g_h=0$), it preempts $\tau_i$ for a duration of $C_h$. If $\tau_h$ uses the GPU ($\eta^g_h>0$), it preempts $\tau_i$ for its entire job execution of $C_h+G^*_h$ which includes the two times of runlist updates in the worst case. 

For case (ii), only GPU-using higher-priority tasks ($\eta^g_h>0$) can affect $\tau_i$. Under busy waiting mode, such $\tau_h$ effectively preempts $\tau_i$ even if $\tau_i$ is a CPU-only task, as exemplified in Fig.~\ref{fig:preemption_cpu_task_dft_prio}. The duration of preemption equals $G^{e*}_h$, which includes the pure GPU execution and the runlist update cost of $2\epsilon$ for each of $\tau_h$'s GPU segments. A release jitter of $J^g_h$ has to be considered here to capture the carry-in effect.
\end{proof}
 
\begin{lemma}\label{lm:baseline_ioctl_rt_suspend}
Under the IOCTL-based approach with self-suspension, the worst-case response time of $\tau_i$ is bounded by:
\begin{equation}
\small
\begin{aligned}
    R_i &= C_i + G^*_i +  (\eta^g_i+1)\cdot\epsilon \\[-2pt]
        &+ \sum_{\substack{\tau_h \in hpp(\tau_i) \\\land \eta^g_h=0}} \lceil \frac{R_i}{T_h} \rceil \cdot C_h 
        + \sum_{\mathclap{\substack{\tau_h \in hpp(\tau_i) \\\land \eta^g_h>0}}} \quad \lceil \frac{R_i + J^c_h}{T_h} \rceil \cdot (C_h + G_h^{m*}) \\[-3pt]
        &+ \sum_{\mathclap{\substack{\tau_h\in hpp(\tau_i) \\\land \eta^g_h>0 \land \eta^g_i>0}}} \quad \lceil \frac{R_i+J^g_h}{T_h} \rceil \cdot G^e_h 
        + \sum_{\mathclap{\substack{\tau_h\in hp(\tau_i)\\\land\tau_h\notin hpp(\tau_i) \\\land \eta^g_h>0 \land \eta^g_i>0}}} \quad \lceil \frac{R_i+J^g_h}{T_h} \rceil \cdot G_h^{e*} 
\end{aligned}
\label{eq:baseline_ioctl_rt_suspend}
\end{equation}
where $J^c_h=R_h-(C_h+G^m_h)$.
\end{lemma}

\begin{proof}
This is a variant of the analysis for busy-waiting tasks given by Lemma~\ref{lm:baseline_ioctl_rt_busy}. The main difference lies in the last three terms accounting for interference from GPU-using higher-priority tasks $\tau_h$. First, $\tau_h$ on the same core as $\tau_i$. Such $\tau_h$ can interfere with $\tau_i$ on both CPU and GPU. On the CPU side (third last term), each job of $\tau_h$ imposes a delay of up to $C_h+G_h^{m*}$ on $\tau_i$ and the self-suspending behavior of $\tau_h$ introduces a jitter effect, $J^c_h$, as reported in~\cite{Bletsas2018}. On the GPU side (second last term), $\tau_h$ imposes interference of its pure GPU execution, $G_h^{e*}$, to $\tau_i$ only when $\tau_i$ uses the GPU ($\eta_i^g>0$). Here, from the perspective of $\tau_i$ with $\eta_i^g>0$, the runlist update delay on the CPU and GPU overlaps, and thus using $G_h^e$ safely bounds GPU preemption from $\tau_h$. Similar to Lemma~\ref{lm:baseline_ioctl_rt_busy}, the carry-in effect can be accounted for by $J^g_h$.

Second, $\tau_h$ on different cores (last term). Such $\tau_h$ interferes with the GPU execution of $\tau_i$ only on the GPU side, by up to $G_h^{e*}$ per its job. By capturing all these cases, the lemma bounds the response time.
\end{proof}

\subsection{Analysis for GPU Priority Assignment}
\label{sec:analysis_gpu_prio}
When the GPU priority assignment given in Sec.~\ref{sec:priority_assignment} is used, the amount of preemption due to higher-priority GPU tasks, i.e., $hpp()$ and $hp()$, needs to be revisited. 
Recall that our assignment preserves the relative priority order of GPU segments the same as the CPU priority order for tasks on the same core. The meaning of $hpp()$ therefore remains unchanged. However, $hp()$ needs to be redefined such that it means the set of tasks with higher ``GPU segment'' priorities in the system. This is because any GPU preemptions in the last terms of Eq.~\eqref{eq:baseline_kthread_rt_busy}, \eqref{eq:baseline_ioctl_rt_busy}, \eqref{eq:baseline_ioctl_rt_suspend},  and the runlist updates in Eq.~\eqref{eq:kernel_thread_delay} are now governed by GPU segment priorities. When computing the release jitter $J^x_h$, $R_h$ needs to be replaced with $D_h$ since the WCRT of higher-priority tasks is unknown when applying our GPU priority assignment method. With these simple modifications, our analysis in the previous section can analyze the effect of the GPU priority assignment.

The use of the GPU segment priority assignment is particularly effective in mitigating the scheduling inefficiency of busy-waiting depicted in Fig.~\ref{fig:preemption_cpu_task_dft_prio}. To illustrate, consider Fig.~\ref{fig:preemption_cpu_task_gpu_prio} where $\tau_2$ is assigned a higher GPU priority than $\tau_1$. Consequently, $\tau_3$ no longer experiences the delay from $\tau_1$'s GPU segment, thereby achieving a shorter response time. Our evaluation results in Sec.~\ref{sec:schedulability_experiments} will confirm this claim. 


\subsection{Analysis with Reduced Pessimism}
\label{sec:analysis_improved}
In the above analysis, the total preemption time on task $\tau_i$ caused by higher-priority tasks running on the same core is simply computed by adding up the worst-case preemption time on both CPU and GPU, assuming both types of preemptions occur throughout execution. However, there are two key factors that make this analysis more conservative than necessary:
\begin{itemize}
    \item Both CPU and GPU preemptions are assumed to happen at their full extent.
    \item Under self-suspension mode, a task $\tau_i$ cannot experience full preemption of all CPU and GPU segments of a higher-priority task $\tau_h$ running on the same core.
\end{itemize}

The first factor can be easily observed in Eqs.~\ref{eq:baseline_kthread_rt_busy},~\ref{eq:baseline_ioctl_rt_suspend} and~\ref{eq:baseline_ioctl_rt_busy} as the interval of interest of $R_i$ is always considered when computing the number of local and remote preemptions. Such pessimism is illustrated in Fig.~\ref{fig:overlap2}. For \ding{192} and \ding{193}, the GPU execution of $\tau_2$ should have a chance to execute together with CPU segments of $\tau_1$, but in the baseline analysis, $\tau_1$'s CPU segments are assumed to preempt $\tau_2$'s GPU execution. 
Fig.~\ref{fig:overlap1} shows the case of an actual schedule where both preemptions of \ding{192} and \ding{193} do not exist.

For the second factor, for convenience, let us consider two tasks $\tau_h$ and $\tau_l$ released at the same time on the same CPU core. $\tau_h$ finishes its CPU segment first and starts GPU execution. At this time, $\tau_l$ can begin its CPU segment and there is an inevitable overlap between $\tau_h$'s GPU segment and $\tau_l$'s CPU segment, making full preemption impossible.

The baseline analysis given in Sec.~\ref{section:sched_analysis:baseline_analysis} overestimates the worst-case response time, especially for the IOCTL-based approach that allows concurrent execution of GPU kernels and CPU segments from different tasks. We focus on reducing the pessimism by identifying minimum possible overlaps between segments to shrink the worst-case response time in the recursive analysis form. We establish key definitions used throughout the discussion.

\begin{figure}[t]
\vspace{-10pt}
    \centering
    \begin{subfigure}[t]{\linewidth}
    \centering
        \includegraphics[width=0.6\linewidth]{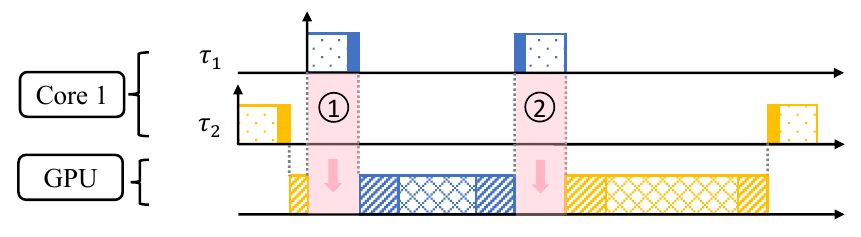}
        \vspace{-0.5\baselineskip}\caption{Assumed pessimistic schedule by baseline analysis. Unreal preemptions are labelled as \ding{192} and \ding{193}.}
        \label{fig:overlap2}
    \end{subfigure}
    \begin{subfigure}[t]{\linewidth}
    \centering
        \includegraphics[width=0.6\linewidth]{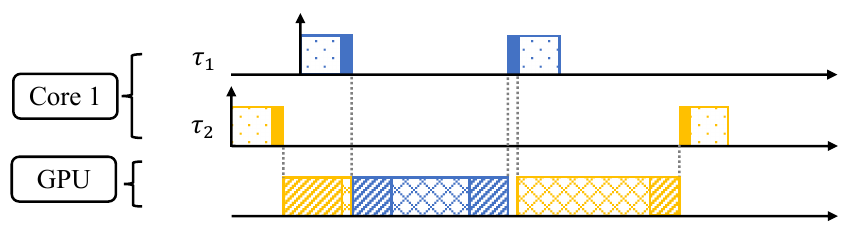}
        \vspace{-0.5\baselineskip}\caption{Actual schedule with overlapped execution}
        \label{fig:overlap1}
    \end{subfigure}
\caption{Pessimism of baseline analysis and two types of execution overlap (IOCTL-based approach; priority: $\tau_1 > \tau_2$)}
\end{figure}

\begin{definition}[Completion time] The completion time, $X$, of a segment is defined as the time interval between the start of execution of the segment and the completion of the segment. 
\label{def:computation_time}
\end{definition}

\begin{definition}[Full overlap] A group of execution segments $e_1$ is said to have a full overlap with another segment $e_2$ if all segments of $e_1$ are entirely contained within $e_2$. This means that the start time $s_1$ of the first segment of $e_1$ is after or equal to the start time $s_2$ of $e_2$ ($s_1 \ge s_2$), and the completion time $c_1$ of the last segment of $e_1$ is before or equal to the completion time $c_2$ of $e_2$ ($c_1 \le c_2$). The notation $e_1 \sqsubset e_2$ denotes that $e_1$ fully overlaps with $e_2$. 
\label{def:full_overlap}
\end{definition}

Based on these, we consider two cases for task $\tau_i$ under analysis: (i) all CPU segments of a higher-priority task $\tau_h$ fully overlap with the $j$-th pure GPU segment of $\tau_i$, i.e., $C_h \sqsubset G_{i,j}^e$, and (ii) all pure GPU segments of $\tau_h$ fully overlap with the $j$-th CPU segment of $\tau_i$, i.e., $G_h^e \sqsubset C_{i,j}$. If the system uses Rate Monotonic or Deadline Monotonic for priority assignment,  lower-priority tasks would tend to have longer periods/deadlines and execution time than those of higher-priority tasks, and finding out the minimum overlapped interval for these cases can yield nontrivial improvements.

\begin{lemma}
The minimum fully overlapped CPU execution of $\tau_h$ with the $j$-th pure GPU segment of $\tau_i$, i.e., $C_h \sqsubset G_{i,j}^e$, is lower-bounded by: 
\begin{equation}
\small
\begin{aligned}
    O^{cg}_{(i,j),h} = \max( (\lfloor \frac{BX^{g}_{i,j}}{T_h} \rfloor - 1) \cdot \widecheck{C_h}, 0)
\end{aligned}
\label{eq:overlap_cg_segment}
\end{equation}
where $BX^{g}_{i,j}$ is the best-case relative completion time of $\tau_i$'s $j$-th pure GPU segment and $\widecheck{C_h}$ is the best-case execution time of all CPU segments of $\tau_h$. $BX^{g}_{i,j}$ is given by:
\begin{equation}
\small
\begin{aligned}
   BX^{g}_{i,j} &= \widecheck{G^e_{i,j}}
                    + \sum_{\tau_h \in hp(\tau_i)}(\lceil \frac{BX^{g}_{i,j}}{T_h} \rceil - 1) \cdot \widecheck{G^e_h}
\end{aligned}
\label{eq:overlap_g_bx}
\end{equation}
where the initial condition for recurrence is $BX^{g}_{i,j}=G_{i,j}^e$, and $\widecheck{G^e_{i,j}}$ is the best-case execution time of the $j$-th pure GPU segment of $\tau_i$.
\label{lemma:overlap_gc}
\end{lemma}
\begin{proof}
The best-case completion time $BX^{g}_{i,j}$ of the $j$-th pure GPU segment of $\tau_i$ given in Eq.~\eqref{eq:overlap_g_bx} is directly adopted from \cite{Bril2004}, the detailed proof of which can be found in that paper.
Next, we determine the minimum number of higher-priority jobs of $\tau_h$ that can fully present in $BX^{g}_{i,j}$. 
Let us use $s^g_{i,j}$ and $c^g_{i,j}$ to denote the absolute start and completion of $\tau_i$'s $j$-th pure GPU segment.
Assuming $m$ arrivals of the higher-priority task $\tau_h$ during $BX^g_{i,j}$ with the start time of the $m$-th job $s_m \le c^g_{i,j}$ and that of the first job $s_1 \ge s^g_{i,j}$, we can deduce that $c_{m-1} - s_1 \le BX^g_{i,j}$ where $c_{m-1}$ is the completion time of $(m-1)$-th job of $\tau_h$, i.e., $c_{m-1}\le s_m$. This indicates that there are at least $m-1$ jobs fully executed within $BX^g_{i,j}$.
We can compute $m$ using $\lfloor \frac{BX^g_{i,j}}{T_h} \rfloor$, and obtain the minimum number of fully overlapped jobs of $\tau_h$, $m-1$, by $\lfloor \frac{BX^g_{i,j}}{T_h} \rfloor - 1$. Thus, Eq.~\eqref{eq:overlap_cg_segment} gives the minimum fully overlapped CPU execution.
\end{proof}
The overall minimum fully overlapped CPU execution of a higher-priority task $\tau_h$ for all pure GPU segments of task $\tau_i$ can be obtained by:
\begin{equation}
\small
\begin{aligned}
    O^{cg}_{i,h} = \sum_{0<j\le\eta^g_i} O^{cg}_{(i,j),h}
\end{aligned}
\label{eq:overlap_cg}
\end{equation}

Eq.~\eqref{eq:overlap_cg} computes the case for $\bigcup C_h \sqsubset G^e_{i,j}$. 
Similarly, for the overlapped execution case for $\bigcup G^e_h \sqsubset C_{i,j}$, we have:
\begin{equation}
\small
\begin{aligned}
    O^{gc}_{i,h} = \sum_{0<j\le\eta^c_i} O^{gc}_{(i,j),h}
\end{aligned}
\label{eq:overlap_gc}
\end{equation}
where
\begin{equation}
\small
\begin{aligned}
    O^{gc}_{(i,j),h} = \max( (\lceil \frac{BX^{c}_{i,j}}{T_h} \rceil - 1) \cdot \widecheck{G^e_h}, 0)
\end{aligned}
\label{eq:overlap_mg_segment}
\end{equation}

Based on the above, we now derive improved analyses. 

\begin{lemma}
\label{lm:improved_rt_busy}
Under the IOCTL-based approach with busy-waiting, the worst-case response time of $\tau_i$ is bounded by: 
\begin{equation}
\small
\begin{aligned}
    &R_i = C_i + G^*_i + (\eta^g_i + 1)\cdot \epsilon \\[-2pt]
        &+ \sum_{\substack{\tau_h \in hpp(\tau_i) \land \eta^g_h=0}} (\lceil \frac{R_i}{T_h} \rceil \cdot C_h - O^{cg}_{i,h})
 \\[-5pt]
        &+ \sum_{\substack{\tau_h \in hpp(\tau_i) \land \eta^g_h>0}} (\lceil \frac{R_i}{T_h} \rceil \cdot (C_h + G^*_h) - (O^{cg}_{i,h} + O^{gc}_{i,h})) \\[-5pt]
        &+ \sum_{\substack{\tau_h\in hp(\tau_i) \land \\ \tau_h\notin hpp(\tau_i)\land\eta^g_h>0}} (\lceil \frac{R_i+J^g_h}{T_h} \rceil \cdot G_h^{e*} - O^{gc}_{i,h})
\end{aligned}
\label{eq:improved_ioctl_rt_busy}
\end{equation}
\end{lemma}
\begin{proof}
As $O^{cg}_{i,h}$ given by Eq.~\eqref{eq:overlap_cg} guarantees the minimum overlapped CPU execution of $\tau_h$ with $\tau_i$'s pure GPU execution, this portion can be safely deducted from the CPU preemption time of $\tau_h$. Similarly, $O^{gc}_{i,h}$ can be deducted from the GPU preemption time.
\end{proof}

\begin{lemma}
\label{lm:improved_rt_suspend}
Under the IOCTL-based approach with self-suspension, the worst-case response time of $\tau_i$ is bounded by: 
\begin{equation}
\small
\begin{aligned}
    &R_i = C_i + G^*_i + (\eta^g_i+1)\cdot\epsilon \\[-2pt]
        &+ \sum_{\substack{\tau_h \in hpp(\tau_i) \land \eta^g_h=0}} (\lceil \frac{R_i}{T_h} \rceil \cdot C_h - O^{cg}_{i,h}) \\[-3pt]
        &+ \sum_{\substack{\tau_h \in hpp(\tau_i) \land \eta^g_h>0}} (\lceil \frac{R_i + J^c_h}{T_h} \rceil \cdot (C_h + G_h^{m*}) - O^{cg}_{i,h}) \\[-3pt]
        &+ \sum_{\mathclap{\substack{\tau_h\in hpp(\tau_i) \\\land \eta^g_h>0 \land \eta^g_i>0}}} \; (\lceil \frac{R_i+J^g_h}{T_h} \rceil \!\cdot\! G_h^e \!-\! O^{gc}_{i,h})
        + \sum_{\mathclap{\substack{\tau_h\in hp(\tau_i)\\\land\tau_h\notin hpp(\tau_i) \\\land \eta^g_h>0 \land \eta^g_i>0}}} \; (\lceil \frac{R_i+J^g_h}{T_h} \rceil \!\cdot\! G_h^{e*} \!-\! O^{gc}_{i,h})
\end{aligned}
\label{eq:improved_ioctl_rt_suspend}
\end{equation}
\end{lemma}

\begin{proof}

The proof directly follows Lemma~\ref{lm:baseline_ioctl_rt_suspend} and~\ref{lm:improved_rt_busy}.
\end{proof}

\section{Evaluation}
We conduct schedulability experiments to compare the proposed approaches against prior work and assess the effect of the GPU priority assignment and the improved analysis. Then, we present a case study on two Nvidia embedded platforms.

\subsection{Schedulability Experiments}
\label{sec:schedulability_experiments}
We generated 1,000 random tasksets for each experimental setting based on the parameters in Table~\ref{tab:params_for_taskset_generation}. The parameter selection is inspired by the prior work~\cite{patel2018analytical}, with slight modifications to increase the system load. Based on the measurement in Sec.~\ref{sec:system_eval} of the manuscript, we aggressively set $\epsilon$ to 1~ms for our approaches, while assuming zero overhead for previous work.
For each task in a taskset, the number of tasks on each CPU is first chosen, and the utilization per CPU is generated based on the UUniFast algorithm~\cite{uunifast}. Then for each task, its period and the number of GPU segments are uniformly randomized within the given range. Then the parameters for each segment are determined. Task priority is assigned by the Rate Monotonic (RM) policy.

\begin{table}[b]
\centering
\begin{tabular}{l|l}
\hline
\textbf{Parameters}                                & \textbf{Value}       \\ \hline
Number of CPUs                                     & 4           \\
Number of tasks per CPU                            & [3, 6]      \\
Ratio of GPU-using tasks                           & [40, 60] \%   \\
Utilization per CPU                                  & [0.4, 0.6] \\
Task Period                                        & [30, 500] ms \\
Number of GPU segments per task                    & [1, 3]      \\
Ratio of GPU exec. to CPU exec. ($G_i/C_i$)        & [0.2, 2]    \\
Ratio of GPU misc. in GPU exec. ($G^m_i/G_i$)      & [0.1, 0.3]  \\
Runlist update cost ($\epsilon$)               & 1 ms  \\ \hline
\end{tabular}
\caption{Parameters for taskset generation}
\label{tab:params_for_taskset_generation}
\end{table}

\subsubsection{Comparison with Prior Work}
We first compare our proposed approaches with two well-known synchronization-based methods, MPCP~\cite{patel2018analytical} and FMLP+~\cite{BB2014-FMLP+}, both of which offer suspension-aware and busy-waiting analyses. For our approaches, we use the improved analysis given in Sec.~\ref{sec:analysis_improved} with the GPU priority assignment in Sec.~\ref{sec:priority_assignment} of the manuscript. Hence, we first run the response time test for a taskset with the default RM priorities, and if the test fails, try again with separate priorities for GPU segments.

We investigate the impact of varying the number of tasks in the taskset, the number of CPUs, the utilization per CPU, and the ratio of GPU-using tasks in Figs.~\ref{fig:grp1:num_of_tasks},~\ref{fig:grp1:num_of_cpus}, ~\ref{fig:grp1:util_per_cpu} and~\ref{fig:grp1:ratio_of_gpu_tasks}, respectively. 
The results show that, in general, the \texttt{ioctl\_busy} and \texttt{ioctl\_suspend} approaches outperform previous methods. However, the \texttt{kthread\_busy} curve occasionally falls below those of previous methods. This is due to that \texttt{kthread\_busy} cannot efficiently utilize computing resources as the system becomes increasingly loaded on the CPU side.

\begin{figure}[]
  \vspace{-10pt}
    \centering
    \includegraphics[width=\linewidth]{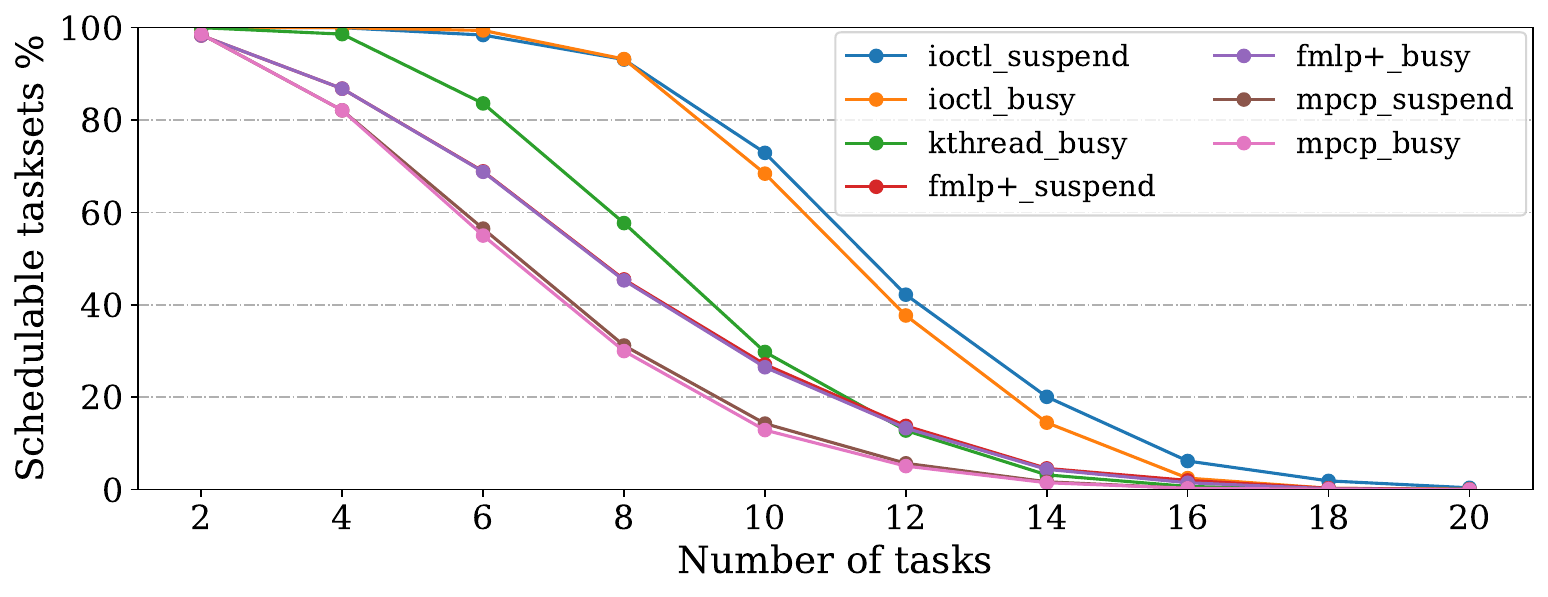}
    \caption{Schedulability w.r.t. the number of tasks}
    \label{fig:grp1:num_of_tasks}
\end{figure}

\begin{figure}[]
  \vspace{-10pt}
    \centering
    \includegraphics[width=\linewidth]{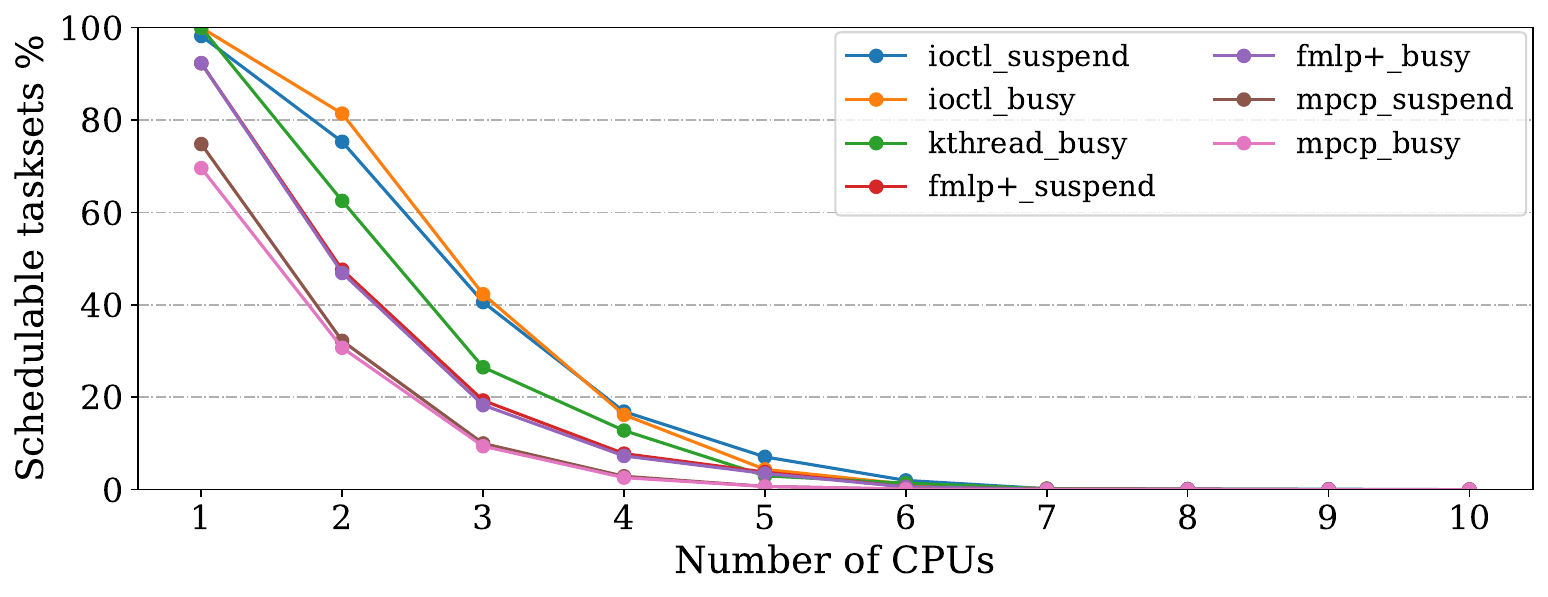}
    \caption{Schedulability w.r.t. the number of CPUs}
    \label{fig:grp1:num_of_cpus}
\end{figure}

\begin{figure}[]
  \vspace{-10pt}
    \centering
    \includegraphics[width=\linewidth]{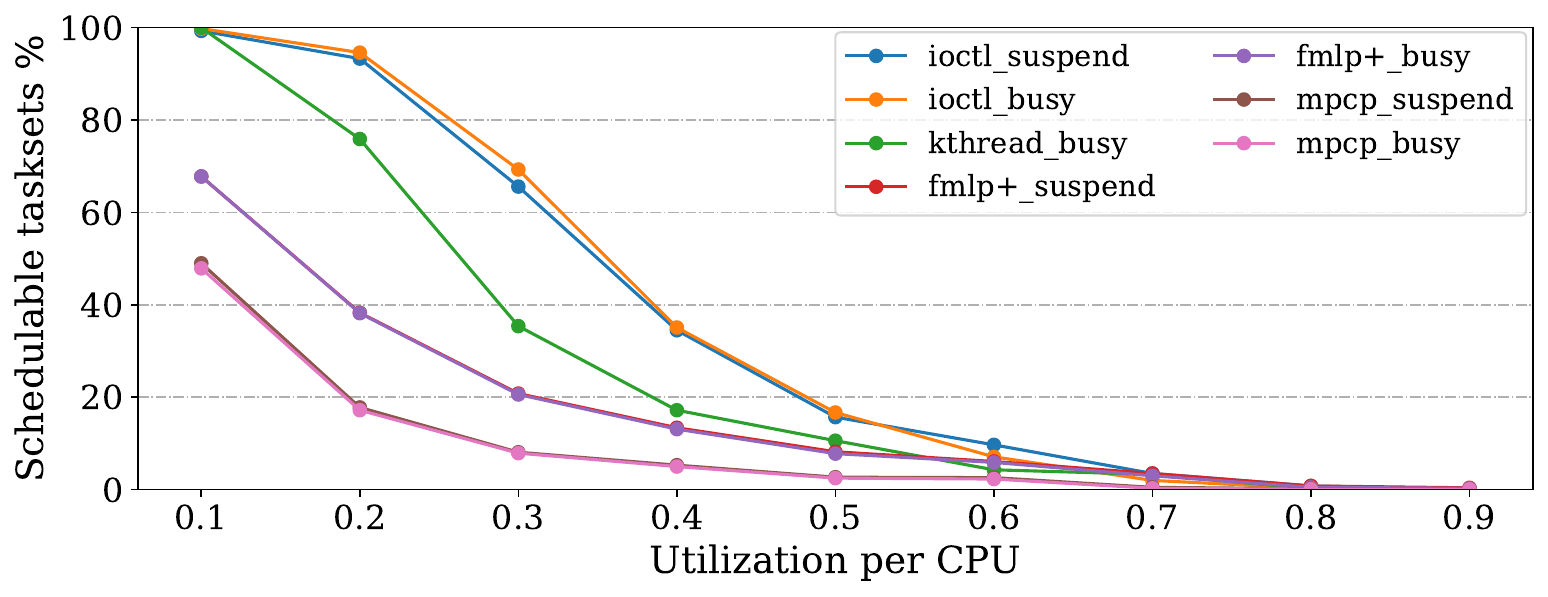}
    \caption{Schedulability w.r.t. the utilization per CPU}
    \label{fig:grp1:util_per_cpu}
\end{figure}

\begin{figure}[]
  \vspace{-10pt}
    \centering
    \includegraphics[width=\linewidth]{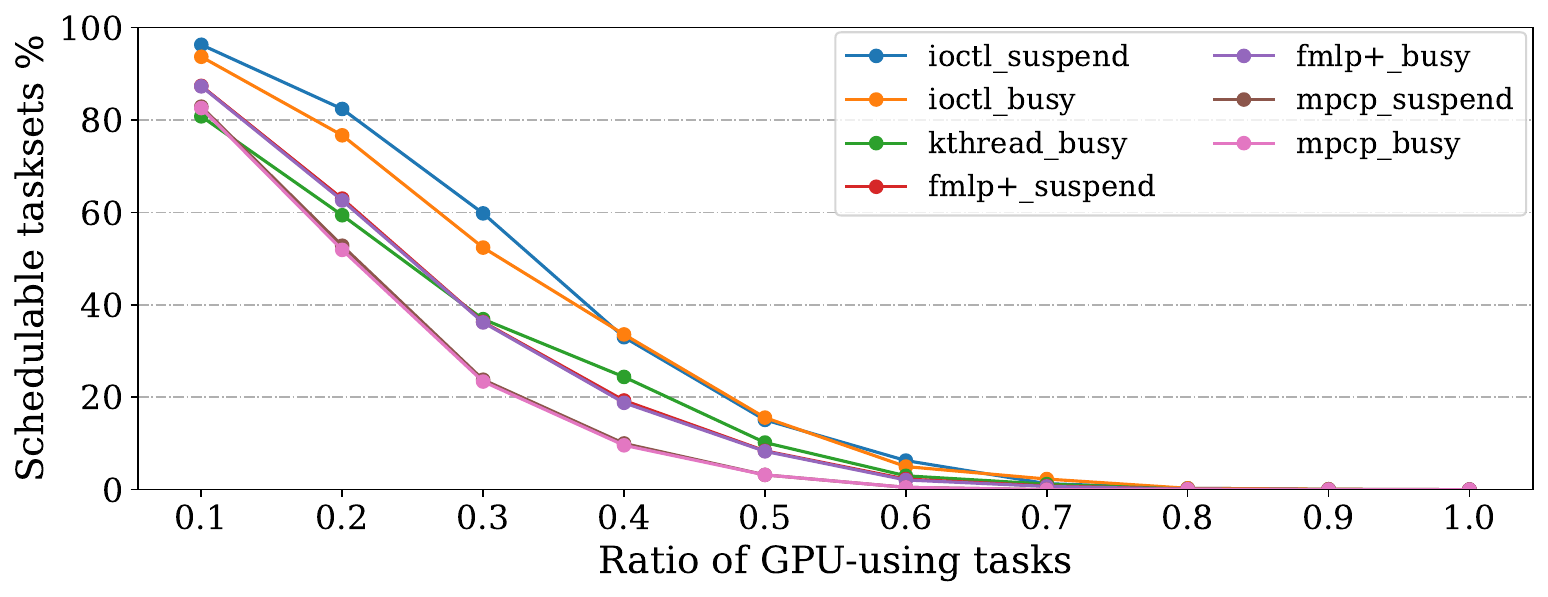}
    \caption{Schedulability w.r.t. the ratio of GPU-using tasks}
    \label{fig:grp1:ratio_of_gpu_tasks}
\end{figure}
\begin{figure}[t]
  \vspace{-10pt}
    \centering
    \includegraphics[width=\linewidth]{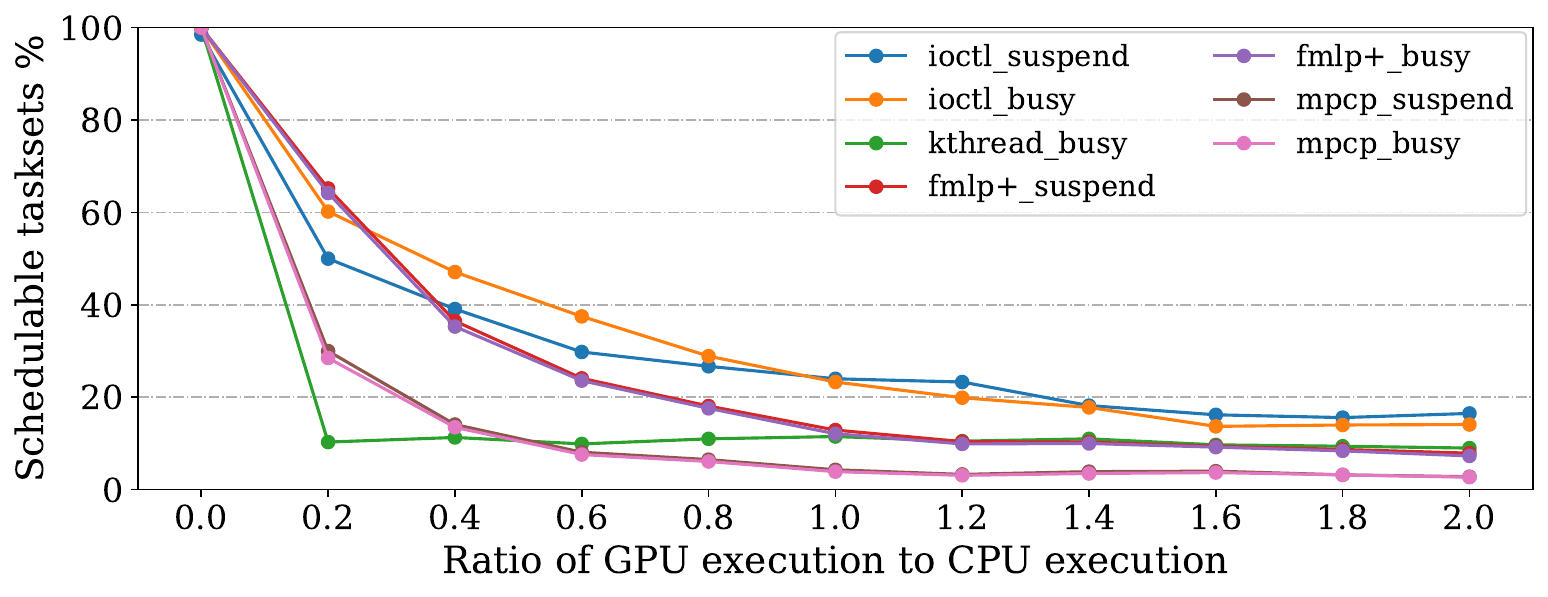}
    \caption{Schedulability w.r.t. the ratio of $G_i$ to $C_i$}
    \label{fig:grp1:ratio_of_g_to_c}
\end{figure}

\begin{figure}[t]
  \vspace{-10pt}
    \centering
    \includegraphics[width=\linewidth]{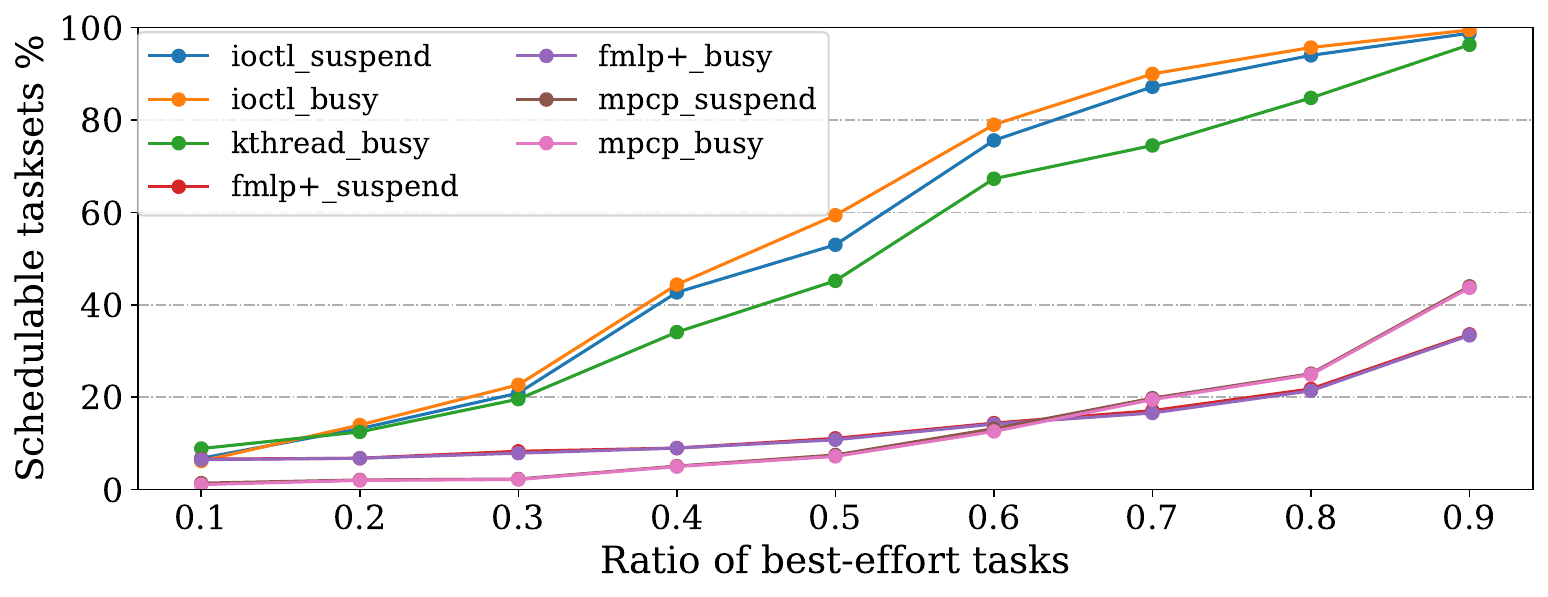}
    \caption{Schedulability w.r.t. the ratio of best-effort tasks}
    \label{fig:grp1:ratio_of_be_tasks}
\end{figure}
Fig.~\ref{fig:grp1:ratio_of_g_to_c} examines the effect of changing the ratio of $G_i/C_i$. When the ratio of $G_i/C_i$ is small, the proposed approaches underperform \texttt{fmlp+} because \texttt{fmlp+} can efficiently schedule tasks when GPU load is light. The advantages of our approaches are mitigated by the critical section of runlist updates, but this trend does not continue as the ratio increases.

Lastly, we explore the impact of best-effort tasks running with the lowest priority in the system.
After generating the tasks using the aforementioned method, we randomly designate a specific percentage of tasks as best-effort tasks in this experiment.
Fig.~\ref{fig:grp1:ratio_of_be_tasks} depicts the percentage of schedulable tasksets as the ratio of best-effort tasks increases. The rest of tasks are all real-time tasks in each taskset with constraint deadlines. The best-effort tasks contribute to blocking time in the analysis of \texttt{mpcp} and \texttt{fmlp+}. Since GPU preemption is enabled in our proposed approaches, they significantly outperform the prior methods.

\subsubsection{Effect of GPU Priority Assignment}
In this experiment, we evaluate the impact of GPU priority assignment on taskset schedulability. We compare baseline analyses of \texttt{kthread\_busy}, \texttt{ioctl\_busy}, and \texttt{ioctl\_suspend} with and without separate GPU priorities, using the taskset generation parameters from Table~\ref{tab:params_for_taskset_generation}.  
Figure~\ref{fig:gain_gpu_prio} illustrates the advantages of GPU priority assignment. Busy-waiting approaches tend to benefit more from this assignment, as explained in Section~\ref{sec:analysis_gpu_prio} of the manuscript. 
Additionally, both busy-waiting and self-suspending approaches benefit from it since assigning GPU priorities independently of CPU priorities makes GPU resources allocation more efficient. E.g. Tasks with shorter GPU segments or higher GPU urgency can be prioritized appropriately, reducing resource wastage.
\begin{figure}[t]
  \vspace{-10pt}
    \centering
    \begin{subfigure}[b]{0.48\linewidth}
        \includegraphics[width=\linewidth]{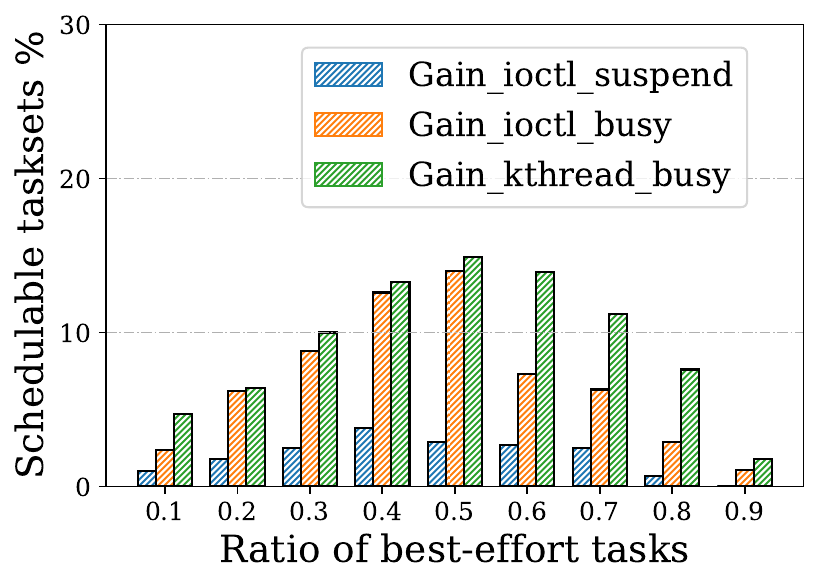}
    \end{subfigure}
    \begin{subfigure}[b]{0.48\linewidth}
        \includegraphics[width=\linewidth]{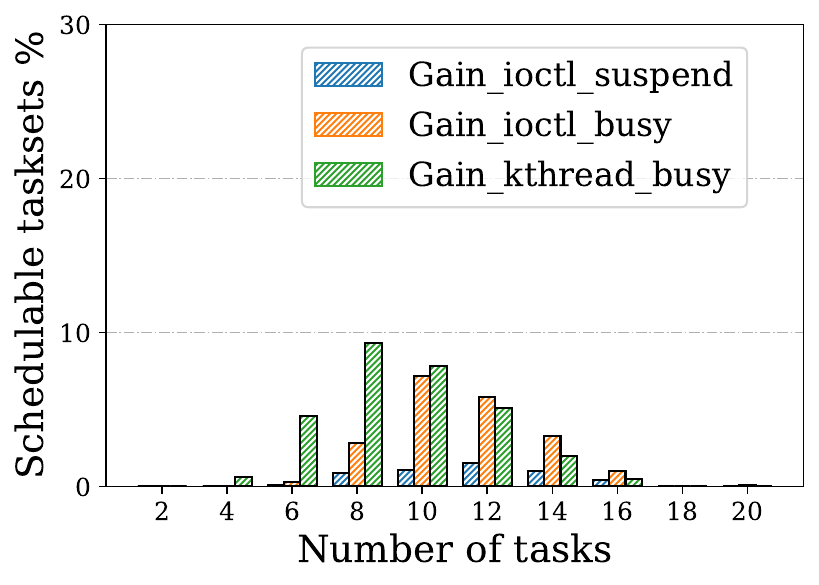}
    \end{subfigure}
    \begin{subfigure}[b]{0.48\linewidth}
        \includegraphics[width=\linewidth]{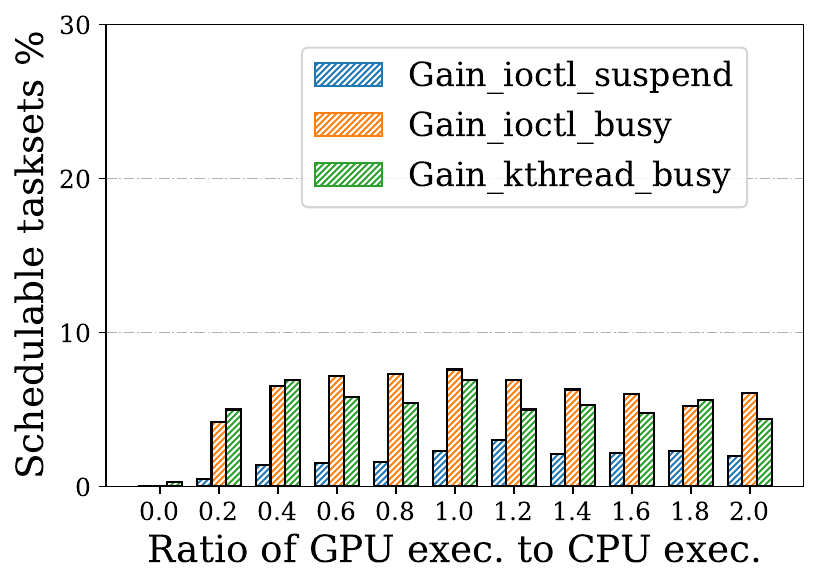}
    \end{subfigure}
    \begin{subfigure}[b]{0.48\linewidth}
        \includegraphics[width=\linewidth]{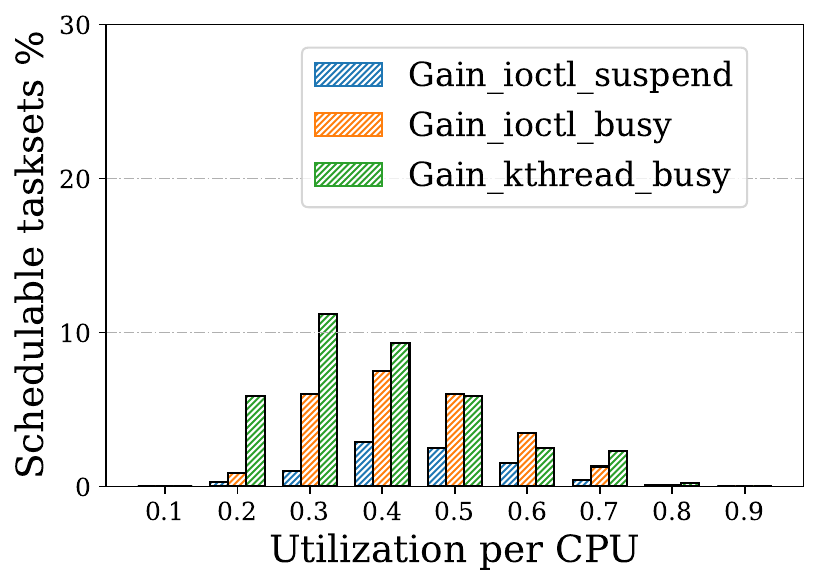}
    \end{subfigure}
    \begin{subfigure}[b]{0.48\linewidth}
        \includegraphics[width=\linewidth]{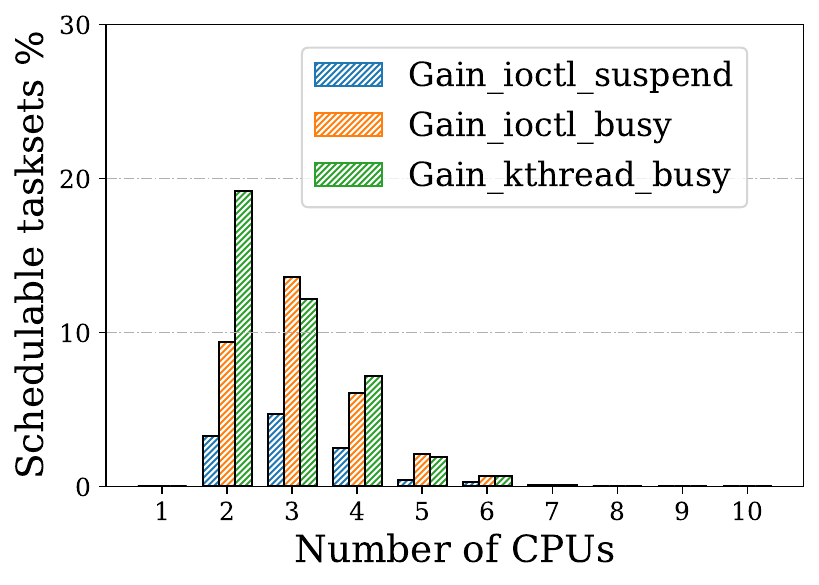}
    \end{subfigure}
    \begin{subfigure}[b]{0.48\linewidth}
        \includegraphics[width=\linewidth]{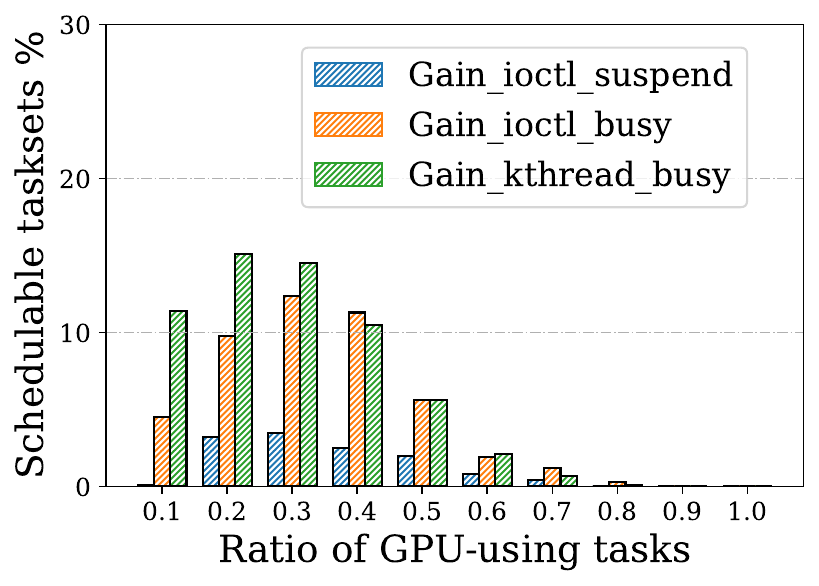}
    \end{subfigure}
\caption{Schedulability gain by GPU priority assignment}
\label{fig:gain_gpu_prio}
\end{figure}

\subsubsection{Effect of Improved Analysis}
Next, we examine the impact of the improved analysis on the overall schedulability of the taskset. 
To better observe the effect, we use the following parameters to generate tasksets: (i) the number of CPUs is set to 2, (ii) two CPU tasks with high utilization and small periods and one GPU task with long GPU execution are always added to the taskset, and (iii) the number of tasks per CPU is [2, 4]. 
All other parameters remain the same as in the previous experiments. Since the improvement does not apply to the kernel thread approach, we only compare the gain under the IOCTL-based approach, i.e., \texttt{ioctl\_busy} and \texttt{ioctl\_suspend}. The results are shown in Fig.~\ref{fig:gain_improved_analysis}.

\begin{figure}[]
    \vspace{-10pt}
    \centering
    \begin{subfigure}[b]{0.48\linewidth}
        \includegraphics[width=\linewidth]{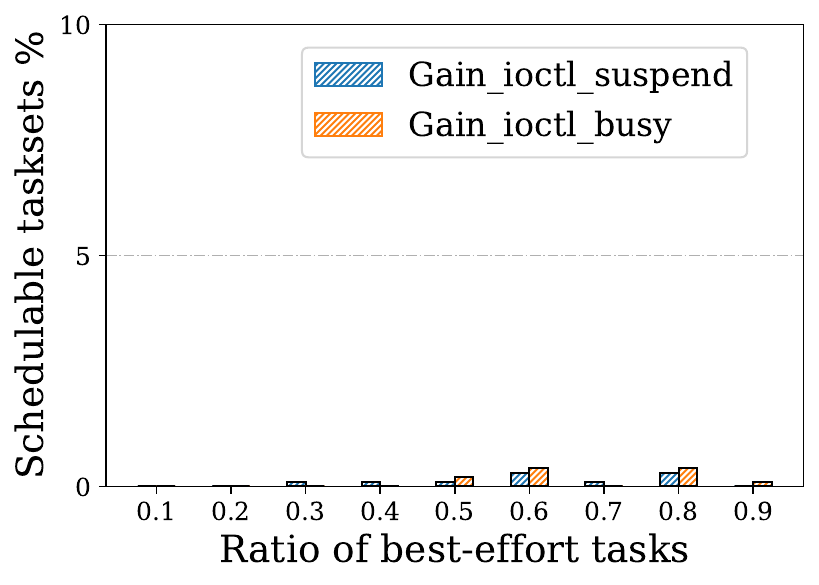}
    \end{subfigure}
    \begin{subfigure}[b]{0.48\linewidth}
        \includegraphics[width=\linewidth]{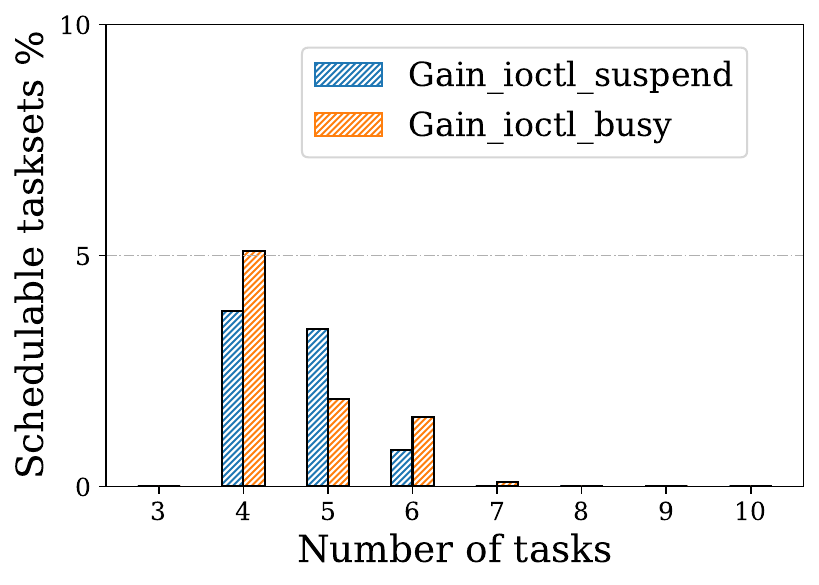}
    \end{subfigure}
    \begin{subfigure}[b]{0.48\linewidth}
        \includegraphics[width=\linewidth]{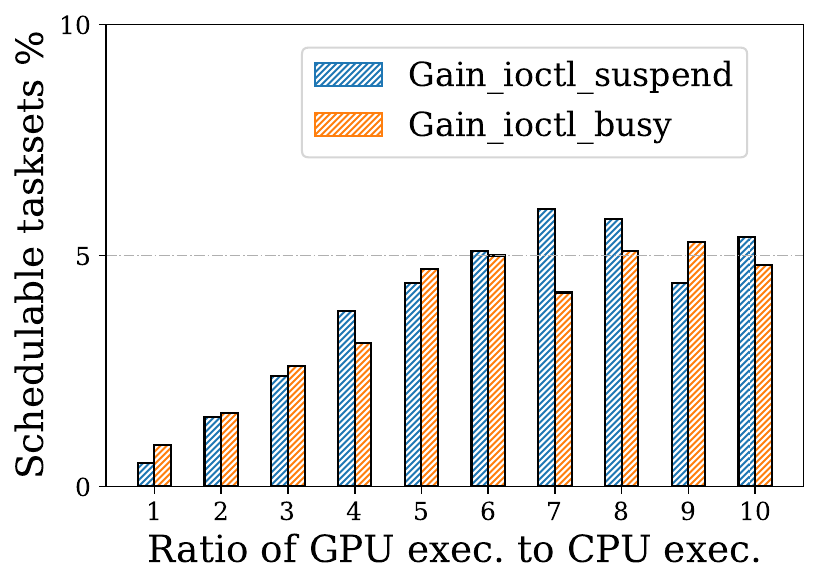}
    \end{subfigure}
    \begin{subfigure}[b]{0.48\linewidth}
        \includegraphics[width=\linewidth]{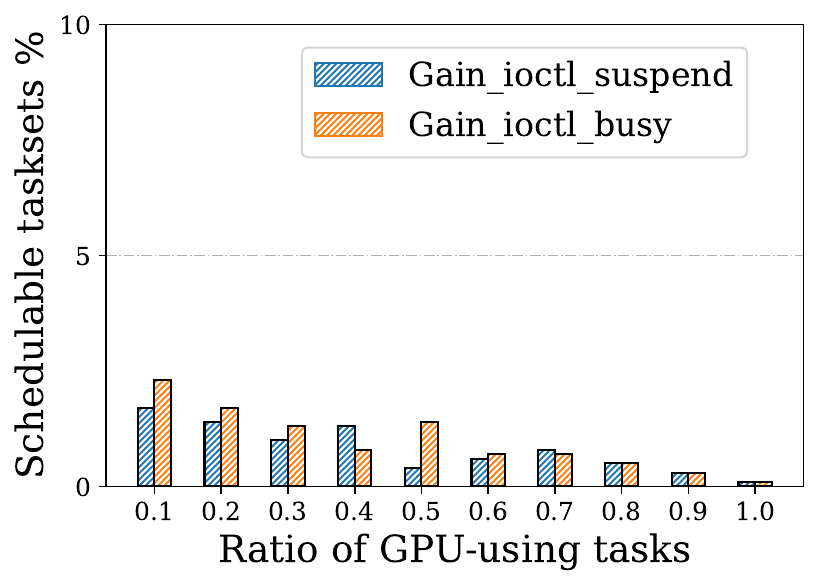}
    \end{subfigure}
\caption{Schedulability gain by improved analysis}
\label{fig:gain_improved_analysis}
\end{figure}

\subsection{System Evaluation}
\label{sec:system_eval}
We implemented our preemptive GPU scheduling approaches on two Nvidia platforms: the Nvidia Jetson Xavier NX Development Kit running L4T R35.1 with Jetpack 5.0.2 and the Nvidia Jetson Orin Nano Developer Kit running L4T R35.4.1 with Jetpack 5.1.2. The first platform features a 6-core 64-bit Carmel ARMv8.2 processor and a Volta architecture GPU. For our experiments, we configured it to run at its highest frequencies in the 6-core 15W mode. The second platform is equipped with a 6-core Arm Cortex-A78AE v8.2 64-bit CPU and an Ampere architecture GPU, and we operated it at its peak frequencies under its default power mode.


\begin{table}[b]
\begin{tabular}{cc|ccccc}
\hline
Task & Workload        & $C_i$ & $G_i$ & $T_i=D_i$ & CPU & Priority \\ \hline
1    & histogram       & 1    & 10   & 100       & 1   & 70       \\
2    & mmul\_gpu\_1    & 2    & 12   & 150       & 2   & 69       \\
3    & mmul\_cpu    & 67   & 0    & 200       & 2   & 68       \\
4    & projection            & 12   & 15   & 300       & 1   & 67       \\
5    & dxtc            & 2    & 16   & 400       & 1   & 66       \\
6    & mmul\_gpu\_2    & 4    & 44   & 200       & 4   & 0        \\
7    & simpleTexture3D & 4    & 27   & 67        & 4,5   & 0        \\ \hline
\end{tabular}
\caption{Taskset used in case study}
\label{tab:taskset_case_study_nx}
\end{table}

\begin{figure}[t]
\vspace{-10pt}
    \centering
    \includegraphics[width=\linewidth]{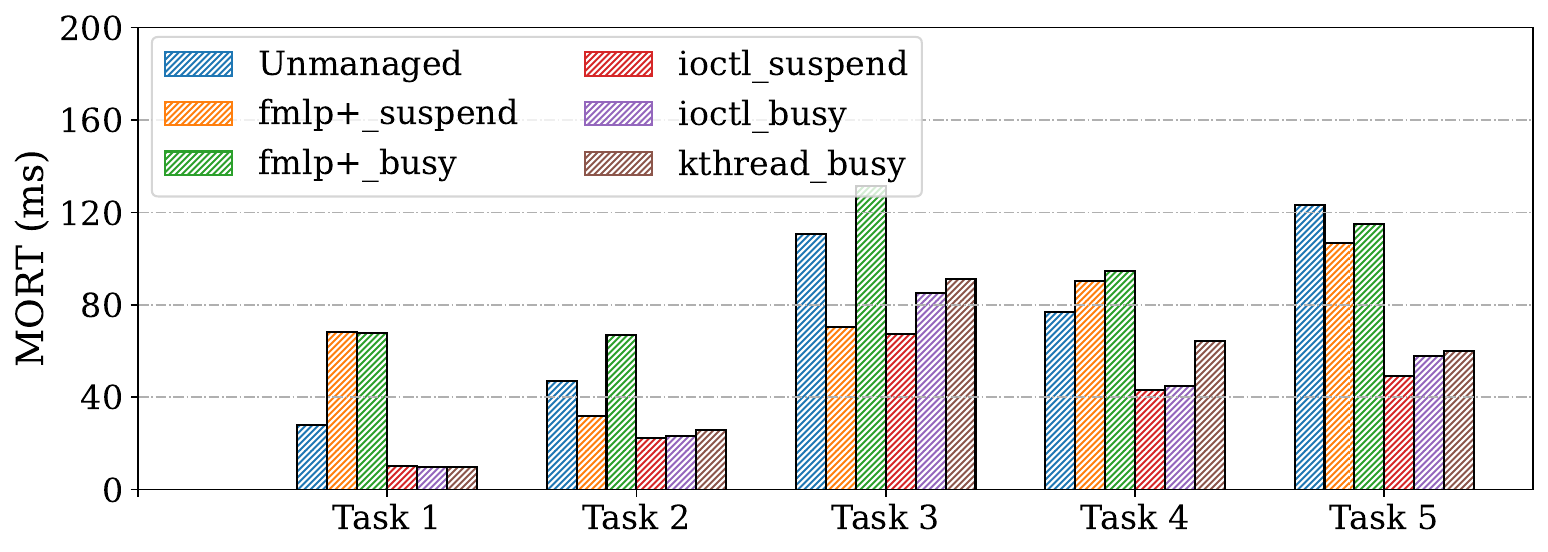}
\vspace{-0.5\baselineskip}     \caption{Maximum observed response time on Jetson Xavier}
    \label{fig:hw_expr:overall_mort_nx}
\end{figure}

\begin{figure}[t]
    \vspace{-10pt}
    \centering
    \begin{subfigure}[b]{0.32\linewidth}
        \includegraphics[width=\linewidth]{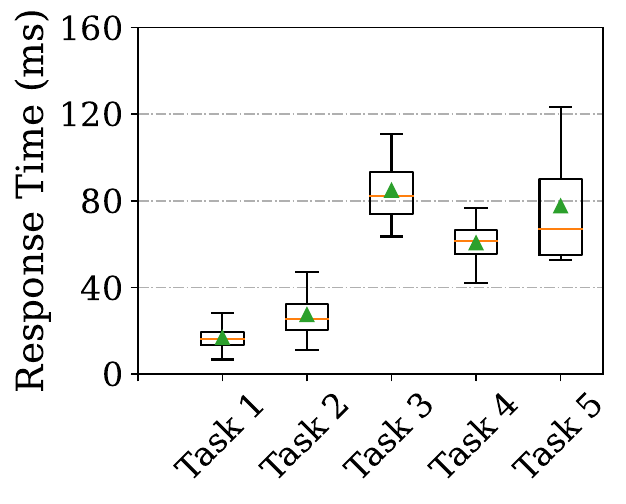}
\vspace{-1.5\baselineskip}\caption{Unmanaged}
    \end{subfigure}
    \begin{subfigure}[b]{0.32\linewidth}
        \includegraphics[width=\linewidth]{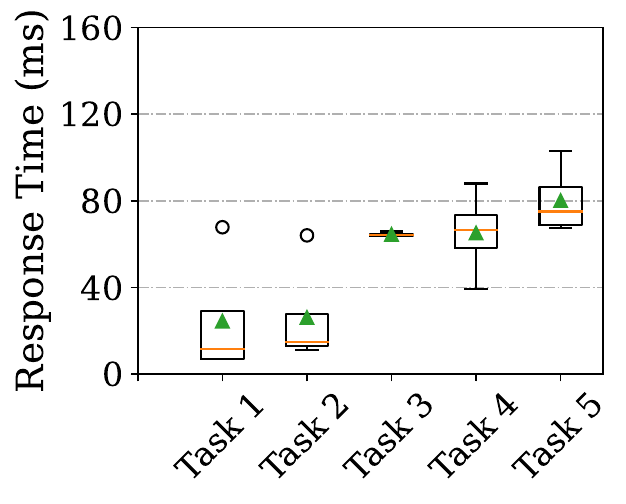}
\vspace{-1.5\baselineskip}\caption{fmlp+\_suspend}
    \end{subfigure}
    \begin{subfigure}[b]{0.32\linewidth}
        \includegraphics[width=\linewidth]{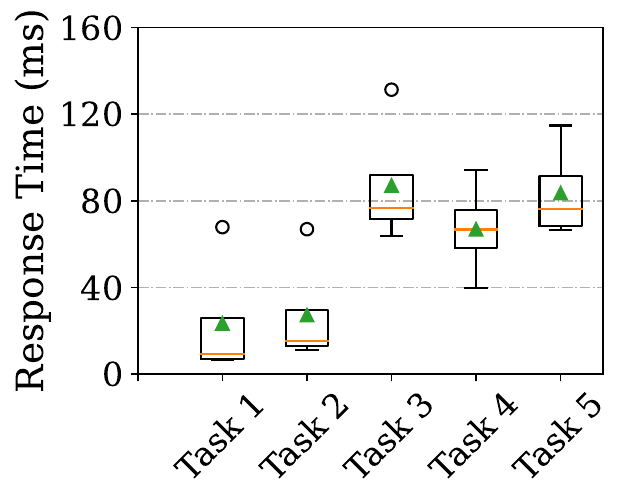}
\vspace{-1.5\baselineskip}\caption{fmlp+\_busy}
    \end{subfigure}
    \begin{subfigure}[b]{0.32\linewidth}
        \includegraphics[width=\linewidth]{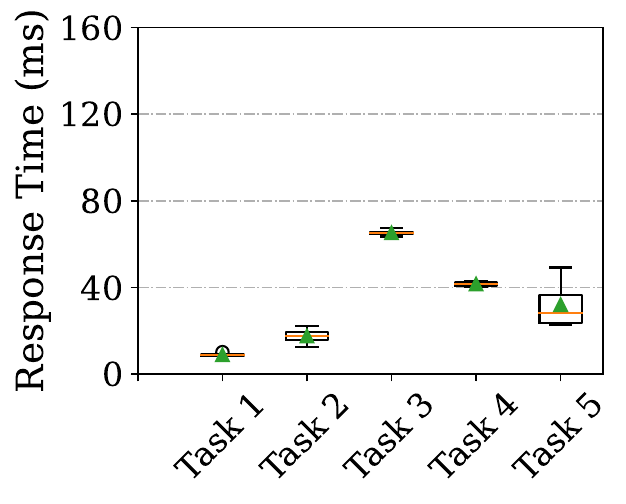}
\vspace{-1.5\baselineskip}\caption{ioctl\_suspend}
    \end{subfigure}
    \begin{subfigure}[b]{0.32\linewidth}
        \includegraphics[width=\linewidth]{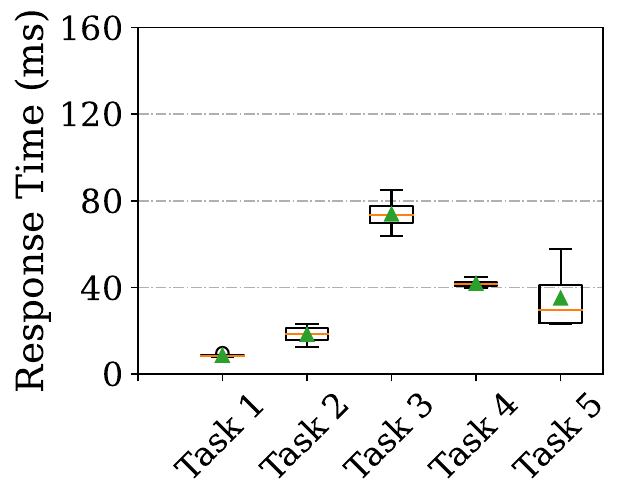}
\vspace{-1.5\baselineskip}\caption{ioctl\_busy}
    \end{subfigure}
    \begin{subfigure}[b]{0.32\linewidth}
        \includegraphics[width=\linewidth]{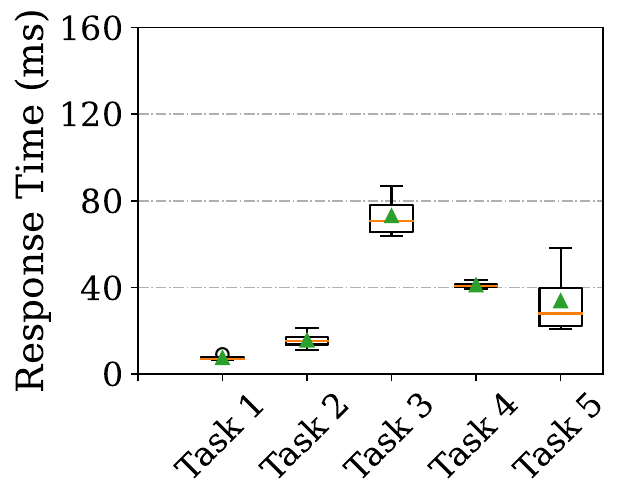}
\vspace{-1.5\baselineskip}\caption{kthread\_busy}
    \end{subfigure}
\caption{Observed response time variations on Jetson Xavier}
\label{fig:hw_expr:individual_mort_nx}
\end{figure}

\begin{figure}[t]
\vspace{-10pt}
    \centering
    \includegraphics[width=\linewidth]{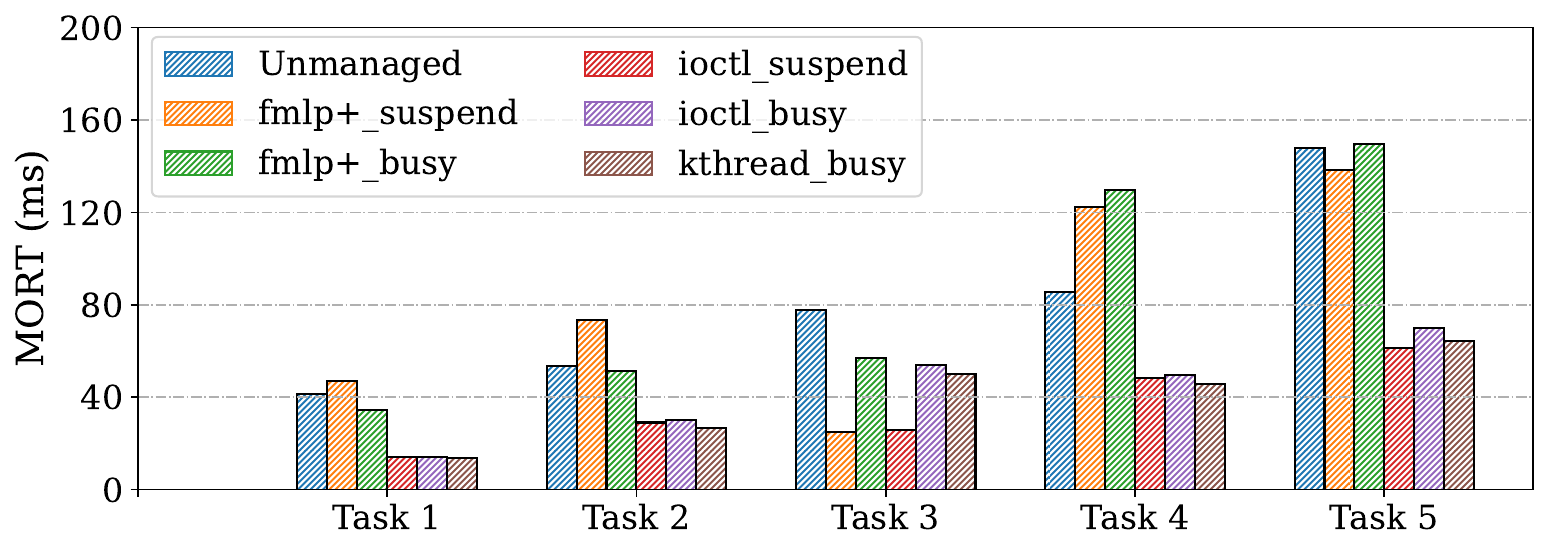}
\vspace{-0.5\baselineskip}     \caption{Maximum observed response time on Jetson Orin}
    \label{fig:hw_expr:overall_mort_orin}
\end{figure}

\begin{table}[t]
\centering
\begin{tabular}{c|cc|cc|cc}
\hline
\multirow{2}{*}{Task} & \multicolumn{2}{c|}{ioctl\_suspend} & \multicolumn{2}{c|}{ioctl\_busy} & \multicolumn{2}{c}{kthread\_busy} \\ \cline{2-7} 
                      & MORT              & WCRT            & MORT            & WCRT           & MORT             & WCRT           \\ \hline
1                     & 10.15             & 16              & 9.68            & 16             & 9.15             & 14             \\
2                     & 22.36             & 28              & 23.28           & 28             & 21.22            & 34             \\
3                     & 67.39             & 73              & 85.01           & 103            & 86.92            & 96             \\
4                     & 43.17             & 66              & 44.91           & 69             & 43.56            & 65             \\
5                     & 49.24             & 111             & 57.93           & 89             & 58.33            & 87             \\ \hline
\end{tabular}
\caption{Comparison of MORT and WCRT on Jetson Xavier}
\label{tab:mort_wcrt_nx}
\end{table}

\noindent\textbf{Case Study.} 
We conducted a case study on the aforementioned platforms to evaluate the performance and effectiveness of the proposed preemptive GPU scheduling mechanism. Table~\ref{tab:taskset_case_study_nx} provides a summary of the taskset employed in this study, and we show the tasks' WCET collected on Jetson Xavier NX. The tasks in the table are arranged in descending order of priority, and each task's GPU segments use the same OS-level priority as its CPU segments. Tasks 3 is a CPU-only task with $G_i = 0$, while the remaining tasks involve GPU computations. Tasks 6 and 7 are categorized as best-effort tasks, as they are not assigned real-time priority.
For \texttt{ioctl\_suspend}, we used CUDA events with the  \texttt{cudaEventBlockingSync} flag to suspend a task during its GPU execution. We compared our approaches against \texttt{unmanaged} (default Nvidia GPU driver) and \texttt{fmlp+} (synchronization-based approach).

We released the tasks at the same time and executed them for a duration of 30s during which we measured the maximum observed response time (MORT) for each real-time task. The results are depicted in Fig.~\ref{fig:hw_expr:overall_mort_nx}. Under unmanaged GPU scheduling with interleaved execution, the response times become unpredictable. The observed response time of each task is shown in Fig.~\ref{fig:hw_expr:individual_mort_nx}. The proposed approaches result in more consistent response times for real-time higher-priority tasks across all \texttt{ioctl\_suspend}, \texttt{ioctl\_busy}, and \texttt{kthread\_busy} scenarios when compared to \texttt{unmanaged} and \texttt{fmlp+}, as evidenced by more compact boxplots with fewer outliers.



Table~\ref{tab:mort_wcrt_nx} lists the comparison between MORT and WCRT computed with our proposed analysis, and the all the WCRT bounds the corresponding MORT.

We run the same experiments on Nvidia Jetson Orin Nano, an embedded GPU platform with the latest Ampere architecture, and the similar trends of MORT are shown in Fig.~\ref{fig:hw_expr:overall_mort_orin}.

\begin{table}[h!]
\centering
\begin{tabular}{l|llll}
\hline
                 & Max  & Min  & Avg  & Medium \\ \hline
Nvidia Jetson Xavier NX & 1021     & 38      & 421      & 508 \\
Nvidia Jetson Orin Nano & 1149     & 16       & 783     & 952 \\ \hline
\end{tabular}
\caption{Runtime overhead of runlist update (in $\mu$s)}
\label{tab:overhead}
\end{table}

\begin{figure}[h!]
    \vspace{-10pt}
    \centering
    \begin{subfigure}[b]{0.45\linewidth}
        \includegraphics[width=\linewidth]{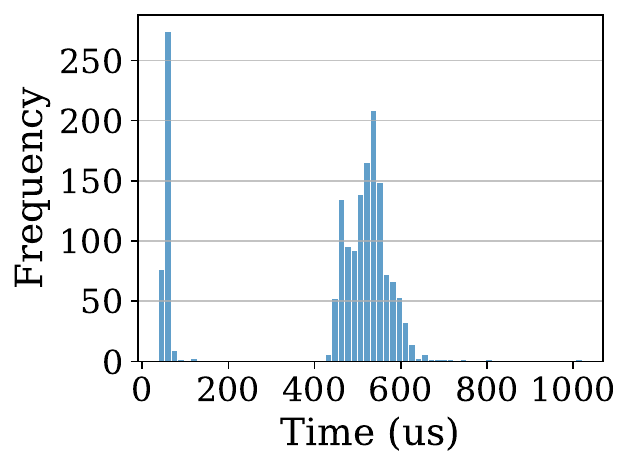}
    \vspace{-1.5\baselineskip}\caption{Jetson Xavier NX}
    \end{subfigure}
    \begin{subfigure}[b]{0.45\linewidth}
        \includegraphics[width=\linewidth]{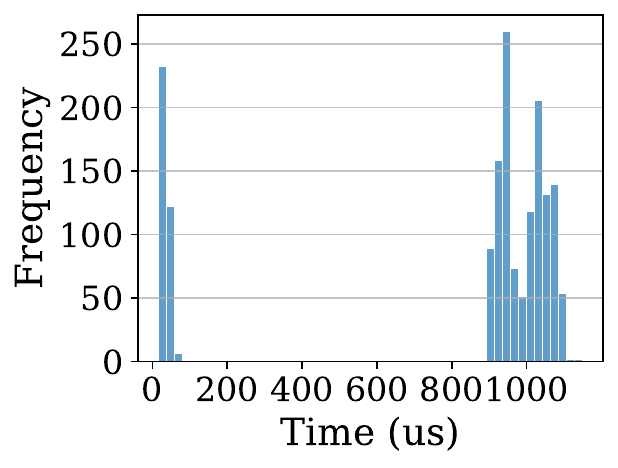}
        \vspace{-1.5\baselineskip}\caption{Jetson Orin Nano}
    \end{subfigure}
\caption{Histogram of runlist update overhead}
\label{fig:overhead_dist}
\end{figure}

\noindent\textbf{Overhead. } We also measured the overhead of runlist update, $\epsilon$, while running the taskset in the case study. The data is shown in Table~\ref{tab:overhead} and the distribution is shown in Fig.~\ref{fig:overhead_dist}. 
The measured overhead includes the cost of IOCTL system call, the cost of scheduling algorithms in kernel space, and the cost of runlist update. The lower mode in the distribution indicates requests that do not necessarily require runlist updates, and it mainly includes the cost of accessing the IOCTL system call. In our experimental settings, these two boards have similar CPU frequencies at about 1.5GHz while Jetson Xavier NX has a much higher GPU frequency of 1.1GHz than Jetson Orin Nano's 625MHz.
Both platforms exhibit a maximum overhead of about 1 ms, which is higher than the range reported in prior work~\cite{capodieci2018deadline}. We suspect this is due to the relatively lower frequency of our GPUs and it could be optimized in future generations of GPU architectures, as can be seen with Orin's case (10\% higher overhead despite half the frequency). Nonetheless, we consider the cost acceptable based on our schedulability experiments conducted with a similar overhead. 
\section{Conclusion}
In this paper, we have proposed two approaches: the IOCTL-based approach and the kernel thread approach to enable preemptive priority-based scheduling of GPU-using tasks in real-time systems. 
We first discussed how the Tegra GPU driver works and presented the design of our approaches. 
Then, we provided a comprehensive response time analysis, with an improvement to reduce pessimism by considering the overlaps between task segments using different computing resources. Through empirical evaluations, we have demonstrated the effectiveness of the proposed approaches in enhancing schedulability, and analyzed the breakdown of the improvements made by the reduced pessimism analysis and GPU priority assignment. Additionally, our case study shows the benefits of our approaches in predictability and responsiveness over the default GPU driver and prior work. The implementation of our approaches will be made publicly available as open source.

Future work can focus on further optimizing and refining the proposed approaches and exploring additional scheduling strategies such as dynamic priority. Combining our device-driver level approaches with GPU partitioning mechanisms will also be an interesting direction.

\input{ref}

\end{document}

%% file: ref.tex
\small
\bibliographystyle{IEEEtran}
\bibliography{ref.bib}

%% file: main.bbl
\begin{thebibliography}{10}
\providecommand{\url}[1]{#1}
\csname url@samestyle\endcsname
\providecommand{\newblock}{\relax}
\providecommand{\bibinfo}[2]{#2}
\providecommand{\BIBentrySTDinterwordspacing}{\spaceskip=0pt\relax}
\providecommand{\BIBentryALTinterwordstretchfactor}{4}
\providecommand{\BIBentryALTinterwordspacing}{\spaceskip=\fontdimen2\font plus
\BIBentryALTinterwordstretchfactor\fontdimen3\font minus \fontdimen4\font\relax}
\providecommand{\BIBforeignlanguage}[2]{{%
\expandafter\ifx\csname l@#1\endcsname\relax
\typeout{** WARNING: IEEEtran.bst: No hyphenation pattern has been}%
\typeout{** loaded for the language `#1'. Using the pattern for}%
\typeout{** the default language instead.}%
\else
\language=\csname l@#1\endcsname
\fi
#2}}
\providecommand{\BIBdecl}{\relax}
\BIBdecl

\bibitem{rajkumar1990real}
R.~Rajkumar, ``Real-time synchronization protocols for shared memory multiprocessors,'' in \emph{Proceedings., 10th International Conference on Distributed Computing Systems}.\hskip 1em plus 0.5em minus 0.4em\relax IEEE Computer Society, 1990, pp. 116--117.

\bibitem{patel2018analytical}
P.~Patel, I.~Baek, H.~Kim, and R.~Rajkumar, ``Analytical enhancements and practical insights for {MPCP} with self-suspensions,'' in \emph{IEEE Real-Time and Embedded Technology and Applications Symposium (RTAS)}, 2018.

\bibitem{BB2014-FMLP+}
B.~B. Brandenburg, ``The fmlp+: An asymptotically optimal real-time locking protocol for suspension-aware analysis,'' in \emph{2014 26th Euromicro Conference on Real-Time Systems}, 2014, pp. 61--71.

\bibitem{Kato2011_RGEM}
S.~{Kato}, K.~{Lakshmanan}, A.~{Kumar}, M.~{Kelkar}, Y.~{Ishikawa}, and R.~{Rajkumar}, ``{RGEM}: A responsive {GPGPU} execution model for runtime engines,'' in \emph{2011 IEEE 32nd Real-Time Systems Symposium}, 2011, pp. 57--66.

\bibitem{Basaran2012}
C.~{Basaran} and K.~{Kang}, ``Supporting preemptive task executions and memory copies in {GPGPUs},'' in \emph{2012 24th Euromicro Conference on Real-Time Systems}, 2012, pp. 287--296.

\bibitem{Zhou2015}
H.~{Zhou}, G.~{Tong}, and C.~{Liu}, ``{GPES}: a preemptive execution system for {GPGPU} computing,'' in \emph{21st IEEE Real-Time and Embedded Technology and Applications Symposium}, 2015, pp. 87--97.

\bibitem{capodieci2018deadline}
N.~Capodieci, R.~Cavicchioli, M.~Bertogna, and A.~Paramakuru, ``{Deadline-based scheduling for GPU with preemption support},'' in \emph{2018 IEEE Real-Time Systems Symposium (RTSS)}.\hskip 1em plus 0.5em minus 0.4em\relax IEEE, 2018, pp. 119--130.

\bibitem{Bakita2023}
J.~Bakita and J.~H. Anderson, ``{Hardware Compute Partitioning on NVIDIA GPUs},'' in \emph{IEEE Real-Time and Embedded Technology and Applications Symposium (RTAS)}, 2023.

\bibitem{xiang2019pipelined}
Y.~Xiang and H.~Kim, ``Pipelined data-parallel {CPU/GPU} scheduling for multi-{DNN} real-time inference,'' in \emph{2019 IEEE Real-Time Systems Symposium (RTSS)}.\hskip 1em plus 0.5em minus 0.4em\relax IEEE, 2019, pp. 392--405.

\bibitem{Elliott2012}
G.~Elliott and J.~Anderson, ``Globally scheduled real-time multiprocessor systems with {GPUs},'' \emph{Real-Time Systems}, vol.~48, pp. 34--74, 05 2012.

\bibitem{Elliott_RTS13}
------, ``An optimal $k$-exclusion real-time locking protocol motivated by {multi-GPU} systems,'' \emph{Real-Time Systems}, vol.~49, no.~2, pp. 140--170, 2013.

\bibitem{Elliott_RTSS13}
G.~Elliott \emph{et~al.}, ``{GPUSync}: A framework for real-time {GPU} management,'' in \emph{IEEE Real-Time Systems Symposium (RTSS)}, 2013.

\bibitem{Saha2019}
S.~Saha, Y.~Xiang, and H.~Kim, ``{STGM}: Spatio-temporal {GPU} management for real-time tasks,'' in \emph{2019 IEEE 25th International Conference on Embedded and Real-Time Computing Systems and Applications (RTCSA)}, 2019, pp. 1--6.

\bibitem{wang2021balancing}
Y.~Wang, M.~Karimi, Y.~Xiang, and H.~Kim, ``{Balancing energy efficiency and real-time performance in GPU scheduling},'' in \emph{2021 IEEE Real-Time Systems Symposium (RTSS)}.\hskip 1em plus 0.5em minus 0.4em\relax IEEE, 2021, pp. 110--122.

\bibitem{Jain2019}
S.~{Jain}, I.~{Baek}, S.~{Wang}, and R.~{Rajkumar}, ``Fractional {GPUs}: Software-based compute and memory bandwidth reservation for {GPUs},'' in \emph{2019 IEEE Real-Time and Embedded Technology and Applications Symposium (RTAS)}, 2019, pp. 29--41.

\bibitem{wang2022towards}
Y.~Wang, M.~Karimi, and H.~Kim, ``{Towards Energy-Efficient Real-Time Scheduling of Heterogeneous Multi-GPU Systems},'' in \emph{2022 IEEE Real-Time Systems Symposium (RTSS)}.\hskip 1em plus 0.5em minus 0.4em\relax IEEE, 2022, pp. 409--421.

\bibitem{zou2023rtgpu}
A.~Zou, J.~Li, C.~D. Gill, and X.~Zhang, ``{RTGPU: Real-time GPU scheduling of hard deadline parallel tasks with fine-grain utilization},'' \emph{IEEE Transactions on Parallel and Distributed Systems}, 2023.

\bibitem{wu2015enabling}
B.~Wu, G.~Chen, D.~Li, X.~Shen, and J.~Vetter, ``{Enabling and exploiting flexible task assignment on GPU through SM-centric program transformations},'' in \emph{Proceedings of the 29th ACM on International Conference on Supercomputing}, 2015, pp. 119--130.

\bibitem{Han2022_reef}
\BIBentryALTinterwordspacing
M.~Han, H.~Zhang, R.~Chen, and H.~Chen, ``Microsecond-scale preemption for concurrent {GPU-accelerated} {DNN} inferences,'' in \emph{16th USENIX Symposium on Operating Systems Design and Implementation (OSDI 22)}.\hskip 1em plus 0.5em minus 0.4em\relax Carlsbad, CA: USENIX Association, Jul. 2022, pp. 539--558. [Online]. Available: \url{https://www.usenix.org/conference/osdi22/presentation/han}
\BIBentrySTDinterwordspacing

\bibitem{brandenburg2014fmlp+}
B.~B. Brandenburg, ``{The FMLP+: An asymptotically optimal real-time locking protocol for suspension-aware analysis},'' in \emph{2014 26th Euromicro Conference on Real-Time Systems}.\hskip 1em plus 0.5em minus 0.4em\relax IEEE, 2014, pp. 61--71.

\bibitem{HKim2017}
H.~{Kim}, P.~{Patel}, S.~{Wang}, and R.~R. {Rajkumar}, ``A server-based approach for predictable {GPU} access control,'' in \emph{2017 IEEE 23rd International Conference on Embedded and Real-Time Computing Systems and Applications (RTCSA)}, 2017, pp. 1--10.

\bibitem{nvidia_preemption}
{AnandTech}, ``{The NVIDIA GeForce GTX 1080 \& GTX 1070 Founders Editions Review},'' \url{https://www.anandtech.com/show/10325/the-nvidia-geforce-gtx-1080-and-1070-founders-edition-review}.

\bibitem{Audsley2007OPTIMALPA}
N.~C. Audsley, ``Optimal priority assignment and feasibility of static priority tasks with arbitrary start times,'' 2007.

\bibitem{bertogna2008schedulability}
M.~Bertogna, M.~Cirinei, and G.~Lipari, ``Schedulability analysis of global scheduling algorithms on multiprocessor platforms,'' \emph{IEEE Transactions on parallel and distributed systems}, vol.~20, no.~4, pp. 553--566, 2008.

\bibitem{Bletsas2018}
K.~Bletsas, N.~C. Audsley, W.-H. Huang, J.-J. Chen, and G.~Nelissen, ``Errata for three papers (2004-05) on fixed-priority scheduling with self-suspensions,'' \emph{Leibniz Transactions on Embedded Systems}, vol.~5, no.~1, p. 02:1–02:20, May 2018.

\bibitem{Bril2004}
R.~Bril, E.~Steffens, and W.~Verhaegh, ``Best-case response times and jitter analysis of real-time tasks,'' \emph{J. Scheduling}, vol.~7, pp. 133--147, 03 2004.

\bibitem{uunifast}
\BIBentryALTinterwordspacing
E.~Bini and G.~C. Buttazzo, ``Measuring the performance of schedulability tests,'' \emph{Real-Time Syst.}, vol.~30, no. 1–2, p. 129–154, may 2005. [Online]. Available: \url{https://doi.org/10.1007/s11241-005-0507-9}
\BIBentrySTDinterwordspacing

\end{thebibliography}
